\documentclass[10pt,journal]{IEEEtran}

\usepackage{amsthm}
\usepackage{amsmath}
\usepackage{slashbox}
\usepackage{graphicx}
\usepackage{epstopdf}
\usepackage{epsfig,graphics,subfigure,psfrag,amsmath,amssymb}
\usepackage{cite}
\usepackage{float}
\usepackage{tocvsec2}
\usepackage{color}
\usepackage{verbatim}
\usepackage{algorithm, algorithmic}
\usepackage{mathtools}
\usepackage{hhline}
\usepackage{enumerate}
\usepackage{arydshln}
\usepackage[bbgreekl]{mathbbol}
\usepackage{bbm}
\usepackage{stackengine}
\usepackage{mathrsfs}
\usepackage{enumitem}
\usepackage{caption}
\def\delequal{\mathrel{\ensurestackMath{\stackon[1pt]{=}{\scriptstyle\Delta}}}}

\newtheorem{thm}{Theorem}
\newtheorem{lem}{Lemma}
\newtheorem{prop}{Proposition}

\newtheorem{defi}{Definition}
\newtheorem{rem}{Remark}

\newtheorem{assump}{Assumption}

\makeatletter

\newcommand{\Rmnum}[1]{\expandafter\@slowromancap\romannumeral #1@}
\makeatother

\def\@#1{\boldsymbol{#1}}
\def\b#1{\mathbb{#1}}
\def\s#1{\mathsf{#1}}
\def\w#1{\widetilde{#1}}
\def\u#1{\underline{#1}}
\def\o#1{\overline{#1}}

\newcommand\ageq{\mathrel{\stackrel{\makebox[0pt]{\mbox{\normalfont\tiny (a)}}}{\geq}}}

\newcommand\aleq{\mathrel{\stackrel{\makebox[0pt]{\mbox{\normalfont\tiny (a)}}}{\leq}}}

\newcommand\aeq{\mathrel{\stackrel{\makebox[0pt]{\mbox{\normalfont\tiny (a)}}}{=}}}

\newcommand\blfootnote[1]{%
  \begingroup
  \renewcommand\thefootnote{}\footnote{#1}%
  \addtocounter{footnote}{-1}%
  \endgroup
}

\begin{document}
\title{\huge{Mobile Edge Caching: An Optimal Auction Approach}}
\author{\normalsize{$^\dagger$Xuanyu Cao, $^\ddag$Junshan Zhang, and $^\dagger$H. Vincent Poor\\
$^\dagger$Email: \{x.cao, poor\}@princeton.edu\\
Department of Electrical Engineering, Princeton University, Princeton, NJ\\
$^\ddag$Email: junshan.zhang@asu.edu\\
School of Electrical, Computer and Energy Engineering, Arizona State University, Tempe, AZ}}
\maketitle
\begin{abstract}
With the explosive growth of wireless data, the sheer size of the mobile traffic is challenging the capacity of current wireless systems. To tackle this challenge, mobile edge caching has emerged as a promising paradigm recently, in which the service providers (SPs) prefetch some popular contents in advance and cache them locally at the network edge. When requested, those locally cached contents can be directly delivered to users with low latency, thus alleviating the traffic load over backhaul channels during peak hours and enhancing the quality-of-experience (QoE) of users simultaneously. Due to the limited available cache space, it makes sense for the SP to cache the most profitable contents. Nevertheless, users' true valuations of contents are their private knowledge, which is unknown to the SP in general. This information asymmetry poses a significant challenge for effective caching at the SP side. Further, the cached contents can be delivered with different quality, which needs to be chosen judiciously to balance delivery costs and user satisfaction. To tackle these difficulties, in this paper, we propose an optimal auction mechanism from the perspective of the SP. In the auction, the SP determines the cache space allocation over contents and user payments based on the users' (possibly untruthful) reports of their valuations so that the SP's expected revenue is maximized. The advocated mechanism is designed to elicit true valuations from the users (incentive compatibility) and to incentivize user participation (individual rationality). In addition, we devise a computationally efficient method for calculating the optimal cache space allocation and user payments. We further examine the optimal choice of the content delivery quality for the case with a large number of users and derive a closed-form solution to compute the optimal delivery quality. Finally, extensive simulations are implemented to evaluate the performance of the proposed optimal auction mechanism, and the impact of various model parameters is highlighted to obtain engineering insights into the content caching problem.\blfootnote{This work was supported by the U.S. Army Research Office under Grant W911NF-16-1-0448.}
\end{abstract}

\begin{IEEEkeywords}
Content caching, content delivery quality, mechanism design, optimal auction
\end{IEEEkeywords}

\section{Introduction}

The last decades have witnessed dramatic proliferation of wireless data. It has been projected that the volume of mobile traffic in 2020 will become 1000 times of that in 2010 \cite{osseiran2013foundation}. Compounding the issue of congested wireless networks is the rapid growth in social network traffic, and indeed many people view contents on their mobile devices.  In particular, a content can become viral, in the sense that it has a rapid increase in popularity in a short time. This has resulted in tremendous pressure on wireless service providers (SPs), because moving a large volume of data into and out of the cloud wirelessly requires substantial spectrum resources, and meanwhile may incur annoying latency and jitter.

A key observation is that many content requests by users are highly repetitive, which would lead to numerous redundant transmissions. With this insight, mobile edge caching is emerging as a new paradigm to alleviate the unprecedented demand on network traffic.   Residing on the network edge, mobile edge caching can offer storage resources to mobile users through low-latency wireless connections, facilitating a number of mobile services.  It is forecast that, by caching contents on network edge, up to 35\% traffic on the backhaul can be reduced. Specifically, SPs can prefetch popular contents and cache them locally at the network edge during off-peak hours when the cellular systems have plenty of resources. As such, the traffic burden of the costly backhaul transmissions is significantly mitigated during peak hours while the quality-of-experience (QoE) of users is also enhanced due to the low transmission delay between local caches and users.

Due to its promising prospect, content caching has attracted extensive research interests in the past few years. Many works have been devoted to various resource allocation issues in caching systems to optimize content placement, caching, routing and delivery \cite{abedini2014content,pacifici2016cache,kvaternik2016methodology,vural2017caching,dehghan2017complexity,ioannidis2016adaptive,gregori2016wireless,tadrous2015proactive,tadrous2016optimal,tadrous2016joint,mokhtarian2017flexible}.  Besides, coded caching \cite{maddah2014fundamental,maddah2015decentralized,niesen2017coded,pedarsani2016online,karamchandani2016hierarchical}, fundamental scaling law \cite{golrezaei2014scaling,ji2015throughput,ji2016fundamental,jeon2017wireless,liu2017much,mahdian2017throughput}, as well as information-theoretic issues \cite{timo2017rate,xu2017fundamental,lim2017information} have also been studied extensively in the literature. In addition, different parties in caching systems, such as content providers (CPs), SPs and users, are often non-cooperative organizations or individuals seeking to maximize their own benefits instead of the overall social welfare, which are amenable to mechanism design \cite{nisan1999algorithmic} and game theory \cite{osborne1994course}. This has spurred a recent tide in designing incentive mechanisms and game-theoretic solutions to content caching networks \cite{dai2012collaborative,li2016pricing,liu2017design,pacifici2017distributed,pacifici2016coordinated,hajimirsadeghi2017joint,wang2017milking}.  Surveys on content caching systems are available in \cite{liu2016caching,wang2017integration,glass2017leveraging,tourani2017security}.

In most existing works on caching systems, user demands of centents are assumed to be either deterministically given or known with fixed distributions. However, in practice, rational users often make strategic content access decisions in response to the content caching and pricing schemes defined by the SPs or CPs. The goals of the users are to maximize their own payoffs selfishly based on their valuations of the contents. Moreover, users' true valuations of contents are usually private knowledge unknown to the SPs/CPs. This information asymmetry makes user demands hardly predictable and may compromise the SPs' profits in content caching systems.

In this paper, we are motivated to design incentive mechanisms for strategic users requesting contents from an SP. To alleviate the traffic load during peak hours, the SP may download some contents from CPs during off-peak hours in advance and cache them in the network edge such as small-cell base stations (SBSs). The prefetched contents can enhance the quality-of-experience (QoE) of users due to the reduction of transmission latency and the users, in turn, pay the SP for the QoE improvement. Because of the limited cache/storage capacity, the SP will only prefetch those most profitable contents from their respective CPs. Nevertheless, the SP is not aware of the users' true valuations or willingness-to-pay, which hinders the SP from choosing profitable contents to prefetch. In this paper, we design an optimal auction mechanism \cite{myerson1981optimal} to maximize the SP's profit. The users report their valuations (not necessarily truthfully) to the SP, who decides the amounts of contents to be cached and the payments of users based on the reports. The main contributions of this work are summarized as follows.
\begin{itemize}
\item An optimal auction mechanism is designed for the SP to prefetch contents and charge users, in a way that maximizes his expected revenue while ensuring truthful reports and user participation. Different from existing optimal auction formulations \cite{myerson1981optimal,nadendla2017optimal,cao2015target,chun2013secondary}, in this paper, the auctioneer (SP) has multiple objects (contents) to sell and the transfer of objects (content delivery) incurs costs, which can be controlled by judicious design of delivery quality. The same content can be shared by multiple users and hence introduces coupling.
\item A systematic method of finding the optimal cache space allocation over contents is presented under certain regularity conditions. In addition, based on the special structure of the optimal cache space allocation scheme, we devise a computationally efficient method for calculating the optimal user payments.
\item Based on the proposed optimal auction mechanism for content caching, we further examine the optimal choice of the content delivery quality for the case with a large number of users. A closed-form solution is presented to determine the optimal delivery quality so as to maximize the expected revenue of the SP.
\item Extensive numerical experiments are conducted to evaluate the performance of the proposed optimal auction mechanism. The impact of various model parameters is studied empirically to obtain engineering insights into the content caching problem.
\end{itemize}

Next, we briefly review the related work of this paper.

\subsection{Related Work}
In the literature, many research efforts have been devoted to various resource allocation issues in caching systems. In \cite{abedini2014content}, Abedini and Shakkottai investigated content caching and scheduling in the presence of elastic and inelastic traffic by using queueing theory. In \cite{pacifici2016cache}, Pacifici \emph{et al.} studied bandwidth allocation for peer-to-peer (P2P) caches with a Markov decision process (MDP) formulation and proposed approximately optimal allocation policy to reduce the costly internet traffic. In \cite{kvaternik2016methodology}, decentralized cooperative caching was studied and caching methods adaptive to the real-time content popularity were proposed to minimize energy consumption of the networks. Furthermore, transient data caching was considered in \cite{vural2017caching}, where the tradeoffs between data freshness and multihop communication costs were identified. In addition, the problem of joint caching and routing in networks was investigated in \cite{dehghan2017complexity} and approximate solutions were presented to minimize the access delay of users. Distributed adaptive content placement in caching networks was analyzed in \cite{ioannidis2016adaptive} to minimize routing costs. Content caching in device-to-device (D2D) networks was examined in \cite{gregori2016wireless}, where contents could be cached at either base stations or user terminals.  Real-time proactive caching with uncertain time-varying user demands was investigated in\cite{tadrous2015proactive,tadrous2016optimal,tadrous2016joint} by means of smart pricing or user demand shaping and prediction. In addition, several practical caching algorithms were developed and evaluated by real data from global video content delivery networks (CDNs) in \cite{mokhtarian2017flexible}.

Furthermore, due to the selfishness of agents in caching systems, many works consider from the viewpoint of game theory and mechanism design, which are more closely related to this paper. In \cite{dai2012collaborative}, Dai \emph{et al.} presented Vickrey-Clarke-Groves (VCG) auction mechanisms to maximize the social welfare in a collaborative caching system with multiple SPs trading bandwidth with each other. In \cite{li2016pricing}, Li \emph{et al.} considered a video caching system where a SP leased SBSs to multiple video retailers (VRs). Through leasing the SBSs, the SP made profit while the VRs could provide faster video transmission to their users. A Stackelberg game formulation of the interaction between the SP and the VRs was put forth to maximize the payoffs of the SP and the VRs jointly. Later, Liu \emph{et al.} extended the model to an information asymmetric scenario, in which the SP only knew the distribution of the CPs' popularity among users \cite{liu2017design}. The authors proposed a contract-theoretic approach to maximize the profit of the SP, who is the monopolist of the caching system. In \cite{pacifici2017distributed}, Pacifici and Dan studied distributed caching with mutiple networked SPs (caches), among which contents could be shared between neighbors with certain latency. Invoking cooperative game theory, the authors proposed convergenet scalable algorithms for distributed caching. Analogous caching problem among networked SPs was investigated through the lens of non-cooperative game theory in \cite{pacifici2016coordinated} by the same authors. In addition, strategic caching and pricing decisions of multiple information centric networks (ICNs) were investigated with game theory in \cite{hajimirsadeghi2017joint} while fairness issues in collaborative caching were examined through Nash bargaining in \cite{wang2017milking}.

The rest of this paper is organized as follows. In Section \Rmnum{2}, the system model is formally presented and the auction design problem is formulated. In Section \Rmnum{3}, we design the optimal auction for content caching and examine the optimal choice of content delivery quality. Numerical results are given in Section \Rmnum{4} and we conclude this work in Section \Rmnum{5}.

\section{System Model and Problem Formulation}

Consider a model with one SP, $m$ CPs and $n$ users, as illustrated in Fig. \ref{model}. The SP may be in charge of one or several base stations serving users within some region. The CPs provide (sell) contents to the SP, who transmits them to the interested mobile users. To mitigate the traffic burden during peak hours, the SP may prefetch (download) some contents and store them in its local cache at the base stations during off-peak hours in advance. When requested (typically during peak hours), the prefetched contents at the cache can be delivered to users with low latency, leading to better QoE for users \cite{wang2017integration}. In turn, the users pay the SP, who thus gains more profits. In what follows, we detail the model of the content caching system in Fig. \ref{model} and formulate the mechanism design problem.

\subsection{System Model}

\begin{figure}
  \centering
  \includegraphics[scale=.38]{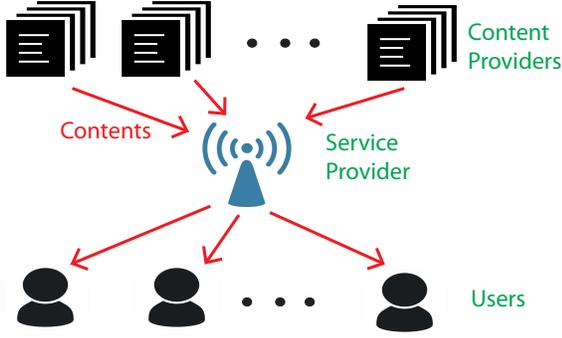}\\
  \caption{An illustration of the caching system under study}\label{model}
\end{figure}

Henceforth, we call the content from CP $i$ content $i$, $i=1,...,m$. Denote the set of contents that interests user $j$ as $\b{S}_j\subset\{1,...,m\}$, $j=1,...,n$. For instance, a Netflix client may be interested in certain categories of movies. On the other hand, denote the set of users interested in content $i$ as $\b{\Omega}_i\subset\{1,...,n\}$. Thus, $i\in\b{S}_j$ is equivalent to $j\in\b{\Omega}_i$. Each user $j$ has a valuation (or willingness-to-pay) $t_j$ of those contents that interest her, i.e., contents in $\b{S}_j$. The valuation $t_j$ is called the type of user $j$ and depends on factors such as the user's demand of contents and economic conditions (richer users tend to have higher valuations or willingness-to-pay). The valuation $t_j$ is only privately known by user $j$ and is unknown to the SP and other users. To model this uncertainty to others, we assume that $t_j$ is a random variable (r.v.) over the interval $\b{T}_j=[\underline{a}_j,\overline{a}_j]$ with probability density function (PDF) $f_j(t_j)$ and cumulative distribution function $F_j(t_j)$. The realization of the r.v. $t_j$ is only revealed to user $j$ while the distribution functions are common knowledge to all users and the SP. The valuations $\@{t}=[t_1,...,t_n]$ are assumed to be independent across users. The joint type space of all users is denoted as $\b{T}=\b{T}_1\times\cdots\times\b{T}_n$ while the joint PDF is denoted as $f(\@{t})=\Pi_{j=1}^nf_j(t_j)$.

To download content $i$, the SP needs to pay CP $i$ with unit price $r_i$, which is set by CP $i$ and considered fixed hereafter. After acquiring contents, the SP delivers them to users as per requests with certain content delivery/transmission quality $\theta>0$, e.g., the video quality of video contents. To deliver contents to $k$ users with quality $\theta$, the SP incurs a unit cost of $kh(\theta)$, where $h:\b{R}_+\mapsto\b{R}_+$ is the cost function of content delivery. The quality parameter $\theta$ is controlled by the SP. If user $j$ receives contents in $\b{S}_j$ with quality $\theta$, his unit satisfaction level will be $\theta t_j$, which depends on both her valuation $t_j$ and the content delivery quality $\theta$.

\subsection{Problem Formulation for Auction Mechanism Design}

Given the limited available cache size at the base stations, the SP aims at downloading the most profitable contents, i.e., those interest many users with high valuations. A challenge faced by the SP is that he does not know the private valuations of users. What the SP knows is only the distribution of the user valuations. Thus, an incentive mechanism is imperative for the SP to elicit truthful valuation information from users and to extract maximal profit. According to the direct revelation principle \cite{myerson1979incentive}, each user $j$ first reports her type (possibly untruthfully) as $\tau_j\in\b{T}_j$, which may deviate from her true type $t_j$. After collecting all the reports $\@{\tau}=[\tau_1,...,\tau_n]\in\b{T}$ from all users, the SP allocates the cache space over contents and charges users according to a mechanism specified by the tuple $\langle\@{p}(\cdot),\@{x}(\cdot)\rangle$. Here, we define $\@{p}(\cdot)=[p_1(\cdot),...,p_m(\cdot)]^\s{T}$ with $p_i:\b{T}\mapsto[0,1]$ prescribing the fraction of caching space used to store content $i$, i.e., $p_i(\@{\tau})$ is the fraction of caching space for content $i$ given the reports $\@{\tau}$. Furthermore, we define $\@{x}(\cdot)=[x_1(\cdot),...,x_n(\cdot)]^\s{T}$ with $x_j:\b{T}\mapsto\b{R}$ prescribing the payment of user $j$, i.e., $x_j(\@{\tau})$ is the payment of user $j$ when the reports are $\@{\tau}$. If all users report their true types, the expected revenue (ER) of the SP is:
\begin{align}
\texttt{ER}\delequal\int_\b{T}\left[\sum_{j=1}^nx_j(\@{t})-\sum_{i=1}^mp_i(\@{t})r_i-\sum_{i=1}^mp_i(\@{t})|\b{\Omega}_i|h(\theta)\right]f(\@{t})d\@{t},\label{er}
\end{align}
where $|\b{\Omega}_i|$ is the cardinality of the set $\b{\Omega}_i$ and the integrands correspond to the user payments, content acquisition costs and content delivery costs, respectively. Given users' reports $\@{\tau}\in\b{T}$ and their true types $\@{t}\in\b{T}$, the ex-post utility of user $j$ under the mechanism is:
\begin{align}
u_j(\@{\tau,t})\delequal\left[\sum_{i\in\b{S}_j}p_i(\@{\tau})\right]\theta t_j-x_j(\@{\tau}).\label{u}
\end{align}
Denote $\b{T}_{-j}=\b{T}_1\times\cdots\times\b{T}_{j-1}\times\b{T}_{j+1}\times\cdots\times\b{T}_n$ as the type space excluding user $j$ and $f_{-j}(\@{t}_{-j})=\Pi_{j'=1,j'\neq j}^nf_{j'}(t_{j'})$ as the joint PDF of types excluding user $j$. Thus, if user $j$ reports $\tau_j$ while his true type is $t_j$ and other users are truth-telling, her interim (expected) utility is:
\begin{align}
v_j(\tau_j,t_j)\delequal\int_{\b{T}_{-j}}u_j(\tau_j,\@{t}_{-j},t_j,\@{t}_{-j})f_{-j}(\@{t}_{-j})d\@{t}_{-j}.\label{v}
\end{align}
Define the truth-telling interim utility of user $j$ as $\w{v}_j(t_j)\delequal v_j(t_j,t_j)$. To incentivize truthful reports of user types, a key step is to guarantee that truth-telling is a Baysian Nash equlibrium \cite{osborne1994course}, i.e.,
\begin{align}
\w{v}_j(t_j)\geq v_j(\tau_j,t_j),~\forall \tau_j,t_j\in\b{T}_j,~\forall j=1,...,n,\label{ic}
\end{align}
which is called the \emph{incentive compatibility (IC)} constraint. In addition, to incentivize users' participation of the auction, the mechanism should satisfy the following (interim) \emph{individual rationality (IR)} constraint:
\begin{align}
\w{v}_j(t_j)\geq0,~\forall t_j\in\b{T}_j,~\forall j=1,...,n,\label{ir}
\end{align}
which means non-negative truth-telling interim utility. Besides, the mechanism obviously needs to satisfy the following \emph{fraction feasibility (FF)} constraint:
\begin{align}
p_i(\@{t})\geq0,~i=1,...,m,~\sum_{i=1}^mp_i(\@{t})\leq1,~\forall \@{t}\in\b{T}.\label{ff}
\end{align}
Overall, the optimization problem faced by the SP is to design optimal auction mechanism $\langle\@{p}(\cdot),\@{x}(\cdot)\rangle$ and optimal delivery quality $\theta$ such that the IC, IR, FF constraints are satisfied and its ER is maximized, i.e.,
\begin{align}
\begin{split}
&\text{Maximize}_{\@{p,x},\theta}~\texttt{ER}\text{ in \eqref{er}}\\
&\text{subject to}~
\begin{cases}
\texttt{FF}\text{ in \eqref{ff}},\\
\texttt{IC}\text{ in \eqref{ic}},\\
\texttt{IR}\text{ in \eqref{ir}},\\
\theta>0.
\end{cases}
\end{split}\label{opt_auc}
\end{align}

\begin{table}
\centering

  \caption{Notations of the model}
  \begin{tabular}{|c|l|}

  \hline
\textbf{Notations}&\textbf{~~~~~~~~~~~~~~~~~~~~Definitions}\\
\hline
$m$&The number of CPs or contents\\
\hline
$n$&The number of users\\
\hline
$\b{S}_j$&The set of contents that interest user $j$\\
\hline
$\b{\Omega}_i$&The set of users interested in content $i$\\
\hline
$\b{T}_j=[\u{a}_j,\o{a}_j]$&The type space of user $j$\\
\hline
$f_j(t_j)$&The PDF of the type of user $j$\\
\hline
$r_i$&The unit price of downloading content $i$\\
\hline
$\theta$&The quality of content delivery\\
\hline
$h(\theta)$&The unit cost of delivering content with quality $\theta$\\
\hline
$\@{p}(\cdot)$&The cache space allocation mechanism\\
\hline
$\@{x}(\cdot)$&The payment mechanism\\
\hline
$u_j(\@{\tau},\@{t})$&The ex-post utility of user $j$\\
\hline
$v_j(\tau_j,t_j)$&The interim utility of user $j$\\
\hline
$\w{v}_j(t_j)$&The truthful interim utility of user $j$\\
\hline
\end{tabular}\label{notations}
\end{table}

There are several differences between the optimal auction problem \eqref{opt_auc} in this paper and the classical Myerson's optimal auction \cite{myerson1981optimal} as well as other applications of optimal auction in various resource allocation problems in \cite{nadendla2017optimal,cao2015target,chun2013secondary}. First, unlike traditional optimal auction problems, the auctioneer (SP) has multiple objects ($m$ contents from different CPs) to sell. Different objects (contents) have different popularity (captured by $\b{\Omega}_i$) among users and the acquisition of different objects incurs different costs (captured by $r_i$). The same content can be shared by multiple users and hence brings forth coupling, i.e., different users may be interested in the same content so that caching it by the SP can benefit multiple users. In fact, exploiting this redundancy in content usage is the keystone of information centric networks and caching systems so that frequently requested contents need not be transmitted repetitively through the costly backhaul channels. Second, different from existing optimal auction frameworks, in this paper, content delivery incurs costs which are controlled by the auctioneer through the delivery quality parameter $\theta$. Increasing $\theta$ boosts users' utilities and thus incentivizes them to pay more. From the perspective of the SP, this enhances its profit extration from users at the expense of higher content delivery costs. Thereby, a judiciously chosen $\theta$ is important to achieve the optimal profit-cost tradeoff. The notations of the model are summarized in Table \ref{notations}.

\section{Optimal Auction Design}

In this section, we solve the optimal auction problem in \eqref{opt_auc}. We first fix the content delivery quality $\theta>0$ and solve for the optimal auction mechanism $\langle\@{p}(\cdot),\@{x}(\cdot)\rangle$. To this end, we transform the IC, IR constraints and the objective function into more tractable forms, and find the optimal mechanism after imposing some regularity conditions on the type distributions (Theorem \ref{sol_fix_theta}). We next show how to compute the optimal payment efficiently (Proposition \ref{prop_payment}) and demonstrate a property of the optimal mechanism (Proposition \ref{property}). Once obtaining the optimal mechanism $\langle\@{p}^*(\cdot),\@{x}^*(\cdot)\rangle$, we further endeavor to find the optimal content delivery quality $\theta$ under the assumption that the number of users $n$ is very large (Theorem \ref{thm_opt_theta}), which is well justified by the prevalence of various content services and mobile networks.

\subsection{Optimal Auction with Fixed Content Delivery Quality $\theta$}

In this subsection, we fix $\theta$ and solve for the optimal $\langle\@{p}(\cdot),\@{x}(\cdot)\rangle$. In other words, for fixed $\theta>0$ we solve the following problem:
\begin{align}
\begin{split}
&\text{Maximize}_{\@{p,x}}~\texttt{ER}\text{ in \eqref{er}}\\
&\text{subject to}~
\begin{cases}
\texttt{FF}\text{ in \eqref{ff}},\\
\texttt{IC}\text{ in \eqref{ic}},\\
\texttt{IR}\text{ in \eqref{ir}}.
\end{cases}
\end{split}\label{opt_auc_fix_theta}
\end{align}

From \eqref{u} and \eqref{v}, we write:
\begin{align}
v_j(\tau_j,t_j)=\theta t_j\w{p}_j(\tau_j)-\w{x}_j(\tau_j),~\forall\tau_j,t_j,~\forall j=1,...,n,\label{v2}
\end{align}
where we have defined:
\begin{align}
&\w{p}_j(\tau_j)\delequal\int_{\b{T}_{-j}}\left[\sum_{i\in\b{S}_j}p_i(\tau_j,\@{t}_{-j})\right]f_{-j}(\@{t}_{-j})d\@{t}_{-j},\label{p_tilde_def}\\
&\w{x}_j(\tau_j)\delequal\int_{\b{T}_{-j}}x_j(\tau_j,\@{t}_{-j})f_{-j}(\@{t}_{-j})d\@{t}_{-j}.
\end{align}
We can transform the IC, IR constraints into equivalent tractable forms according to the following lemma.
\begin{lem}\label{lem_constraints}
Suppose the FF constraint in \eqref{ff} holds. Then, the necessary and sufficient condition for IC in \eqref{ic} and IR in \eqref{ir} is:
\begin{enumerate}[label=(\roman*)]
\item $\w{p}_j(\cdot)$ is an increasing function for each $j=1,...,n$.
\item For any $t_j\in\b{T}_j$, $j=1,...,n$: $\w{v}_j(t_j)=\w{v}_j(\u{a}_j)+\theta\int_{\u{a}_j}^{t_j}\w{p}_j(\tau_j)d\tau_j$.
\item For each $j=1,...,n$: $\w{v}_j(\u{a}_j)\geq0$.
\end{enumerate}
\end{lem}
\begin{proof}
The proof is presented in Appendix A.
\end{proof}

Based on Lemma \ref{lem_constraints}, we can simplify the expected revenue, i.e., the objective function of problem \eqref{opt_auc_fix_theta}, as follows.

\begin{lem}\label{lem_er}
For any feasible mechanism of problem \eqref{opt_auc_fix_theta}, i.e., $\langle\@{p}(\cdot),\@{x}(\cdot)\rangle$ satisfying the FF, IC and IR constraints, the expected revenue becomes:
\begin{align}
\texttt{ER}=&-\sum_{j=1}^n\w{v}_j(\u{a}_j)+\theta\int_\b{T}\Bigg\{\sum_{j=1}^n\left[t_j-\frac{1-F_j(t_j)}{f_j(t_j)}-\frac{h(\theta)}{\theta}\right]\nonumber\\
&\cdot\left(\sum_{i\in\b{S}_j}p_i(\@{t})\right)\Bigg\}f(\@{t})d\@{t}-\int_\b{T}\left[\sum_{i=1}^mp_i(\@{t})r_i\right]f(\@{t})d\@{t}.\label{er_alt}
\end{align}
\end{lem}
\begin{proof}
The proof is presented in Appendix B.
\end{proof}

Combining Lemmas \ref{lem_constraints} and \ref{lem_er}, we see that problem \eqref{opt_auc_fix_theta} is equivalent to:
\begin{align}
\begin{split}
&\text{Maximize}_{\@{p,x}}~\texttt{ER}\text{ in \eqref{er_alt}}\\
&\text{subject to}~
\begin{cases}
\texttt{FF}\text{ in \eqref{ff}},\\
\text{(i), (ii), (iii) in Lemma \ref{lem_constraints}}.
\end{cases}
\end{split}\label{prob_px}
\end{align}
According to definitions, $\w{v}_j(\cdot)$ is related to both $\@{p}(\cdot)$ and $\@{x}(\cdot)$ while $\w{p}_j(\cdot)$ only depends on $\@{p}(\cdot)$ (independent of $\@{x}(\cdot)$). We note that, for ER in \eqref{er_alt}, only the first term $-\sum_{j=1}^n\w{v}_j(\u{a}_j)$ depends on $\@{x}(\cdot)$ and the rest terms are independent of $\@{x}(\cdot)$. Furthermore, constraints FF and (i) in Lemma \ref{lem_constraints} depend on $\@{p}(\cdot)$ only and are independent of $\@{x}(\cdot)$. Therefore, if we fix some $\@{p}(\cdot)$ satisfying constraints FF and (i) of Lemma \ref{lem_constraints} and optimize over $\@{x}(\cdot)$, problem \eqref{prob_px} becomes:
\begin{align}
\begin{split}
&\text{Maximize}_{\@{x}}~-\sum_{j=1}^n\w{v}_j(\u{a}_j)\\
&\text{subject to:}~\text{(ii), (iii) in Lemma \ref{lem_constraints}}.
\end{split}\label{prob_x}
\end{align}

The solution to problem \eqref{prob_x} is given as follows.
\begin{prop}\label{x_opt}
For fixed $\@{p}(\cdot)$ satisfying constraints FF and (i) of Lemma \ref{lem_constraints}, the optimal solution to problem \eqref{prob_x} is:
\begin{align}
x_j^*(\@{t})\delequal\theta t_j\sum_{i\in\b{S}_j}p_i(\@{t})-\theta\int_{\u{a}_j}^{t_j}\sum_{i\in\b{S}_j}p_i(\tau_j,\@{t}_{-j})d\tau_j,\nonumber\\ \forall j=1,...,n,~\forall t\in\b{T},\label{x_sol}
\end{align}
which achieves the optimal value of 0.
\end{prop}
\begin{proof}
According to (iii) in Lemma \ref{lem_constraints}, we see that the optimal value of \eqref{prob_x} is no greater than 0. From \eqref{x_sol}, we know that, for any $\@{t}_{-j}\in\b{T}_{-j}$, $x_j^*(\u{a}_j,\@{t}_{-j})=\theta\u{a}_j\sum_{i\in\b{S}_j}p_i(\u{a}_j,\@{t}_{-j})$. So, (iii) of Lemma \ref{lem_constraints} holds with equality for $\@{x}^*(\cdot)$, i.e., $\w{v}_j(\u{a}_j)=0$ for each $j$. Substituting \eqref{x_sol} into the L.H.S. of (ii) in Lemma \ref{lem_constraints}, interchanging the order of integrals and noting the definition of $\w{p}_j(\cdot)$, we can show that (ii) also holds for $\@{x}^*(\cdot)$. So, $\@{x}^*(\cdot)$ is feasible for problem \eqref{prob_x}. Meanwhile, it also achieves an objective function value of 0. Thus, it is optimal for problem \eqref{prob_x}.
\end{proof}
Next, we optimize over $\@{p}(\cdot)$. According to Proposition \ref{x_opt}, we only need to solve the following problem:
\begin{align}
\begin{split}
&\text{Maximize}_{\@{p}}~\int_\b{T}\Bigg\{\sum_{j=1}^n\left[\theta t_j-\frac{\theta(1-F_j(t_j))}{f_j(t_j)}-h(\theta)\right]\\
&~~~~~~~~~~~~~~\cdot\left[\sum_{i\in\b{S}_j}p_i(\@{t})\right]-\sum_{i=1}^mp_i(\@{t})r_i\Bigg\}f(\@{t})d\@{t}\\
&\text{subject to}~
\begin{cases}
\texttt{FF}\text{ in \eqref{ff}},\\
\text{(i) in Lemma \ref{lem_constraints}}.
\end{cases}
\end{split}\label{prob_p}
\end{align}
Define $c_j(t_j)\delequal t_j-\frac{1-F_j(t_j)}{f_j(t_j)}$, for any $j=1,...,n$, $t_j\in\b{T}_j$. Note that $c_j(\cdot)$ is a function depending only on the distribution of the true type of user $j$. To facilitate solving problem \eqref{prob_p}, we make the following regularity assumption, which is common in the mechanism design literature \cite{myerson1981optimal,nadendla2017optimal,cao2015target}.
\begin{defi}\label{regularity}
[Regularity] The optimal auction problem \eqref{opt_auc_fix_theta} is \emph{regular} if $c_j(\cdot)$ is an increasing function on $\b{T}_j$ for each $j=1,...,n$.
\end{defi}
We remark that this regularity condition holds for many familiar distributions, e.g., uniform distributions and exponential distributions. One may refer to \cite{baron1984regulation} for a partial list. Under the regularity assumption, we can compute the optimal solution to problem \eqref{prob_p} as follows.

\begin{prop}\label{p_opt}
Suppose the optimal auction problem \eqref{opt_auc_fix_theta} is regular. Then, the optimal solution $\@{p}^*(\cdot)$ to problem \eqref{prob_p} is constructed as follows. For each $\@{t}\in\b{T}$, define $k=\arg\max_{i=1,...,m}\left\{\sum_{j\in\b{\Omega}_i}[\theta c_j(t_j)-h(\theta)]-r_i\right\}$, whose value depends on $\@{t}$. Then, the optimal $\@{p}^*(\@{t})$ is given as:
\begin{align}
\begin{cases}
p_k^*(\@{t})=1,~p_i^*(\@{t})=0,~\forall i\neq k,\\ 
~~~~~~~~~\text{if}~\max_{i=1,...,m}\left\{\sum_{j\in\b{\Omega}_i}[\theta c_j(t_j)-h(\theta)]-r_i\right\}>0;\\
p_i^*(\@{t})=0,~\forall i=1,...,m,~\text{otherwise.}\label{p_sol}
\end{cases}
\end{align}
\end{prop}
\begin{proof}
We first omit the constraint (i) of Lemma \ref{lem_constraints} and study the following relaxed version:
\begin{align}
\begin{split}
&\text{Maximize}_{\@{p}}~\text{The objective function of \eqref{prob_p}}\\
&\text{subject to:}~\texttt{FF}\text{ in \eqref{ff}}.
\end{split}\label{prob_p_relax}
\end{align}
Later, we will show that the optimal solution to this relaxed problem happens to satisfy (i) as well, i.e., it is also optimal for the original problem \eqref{prob_p}. We note both the objective and the constraint of problem \eqref{prob_p_relax} are decoupled across types. So, problem \eqref{prob_p_relax} can be solved for each given type $\@{t}\in\b{T}$. For fixed $\@{t}\in\b{T}$, by making use of the definition of $c_j(\cdot)$ and rearranging terms in the objective function, we can rewrite the corresponding problem as:
\begin{align}
\begin{split}
&\text{Maximize}_{\@{p}(\@{t})}~\sum_{i=1}^m\left\{\sum_{j\in\b{\Omega}_i}[\theta c_j(t_j)-h(\theta)]-r_i\right\}p_i(t)\\
&\text{subject to:}~
\begin{cases}
p_i(\@{t})\geq0,~i=1,...,m,\\
\sum_{i=1}^mp_i(\@{t})\leq1.
\end{cases}
\end{split}\label{prob_p_relax_t}
\end{align}
Problem \eqref{prob_p_relax_t} is a simple linear program (LP), whose optimal solution is clearly $\@{p}^*(\@{t})$ given in \eqref{p_sol}. It now remains to show that $\@{p}^*(\cdot)$ satisfies statement (i) in Lemma \ref{lem_constraints}, i.e., $\int_{\b{T}_{-j}}\left[\sum_{i\in\b{S}_j}p_i^*(\@{t})\right]f_{-j}(\@{t}_{-j})d\@{t}_{-j}$ is an increasing function of $t_j\in\b{T}_j$ for every $j\in\{1,...,n\}$, so that the aforementioned relaxation is exact.

To this end, for any $j\in\{1,...,n\}$, consider arbitrary $\tau_j,t_j\in\b{T}_j$ with $\tau_j\leq t_j$. By the regularity assumption, we have $c_j(\tau_j)\leq c_j(t_j)$. From the construction of $\@{p}^*(\cdot)$ in \eqref{p_sol}, we know that $\sum_{i\in\b{S}_j}p_i^*(\tau_j,\@{t}_{-j})$ is either 0 or 1, for any $\@{t}_{-j}\in\b{T}_{-j}$. Next, for arbitrarily fixed $\@{t}_{-j}\in\b{T}_{-j}$, we distinguish two cases.

\emph{Case 1: $\sum_{i\in\b{S}_j}p_i^*(\tau_j,\@{t}_{-j})=1$.} In such a case, there exists some $l\in\b{S}_j$ such that $p_l^*(\tau_j,\@{t}_{-j})=1$. Hence, according to the construction of $\@{p}^*(\cdot)$, we have:
\begin{align}
&\theta c_j(\tau_j)-h(\theta)+\sum_{s\in\b{\Omega}_l,s\neq j}[\theta c_s(t_s)-h(\theta)]-r_l\\
&=\max\Bigg\{\max_{i\in\b{S}_j}\Bigg\{\theta c_j(\tau_j)-h(\theta)+\sum_{s\in\b{\Omega}_i,s\neq j}[\theta c_s(t_s)-h(\theta)]\nonumber\\
&~~~~-r_i\Bigg\},\max_{i\notin\b{S}_j}\left\{\sum_{s\in\b{\Omega}_i}[\theta c_s(t_s)-h(\theta)]-r_i\right\}\Bigg\}\\
&>0.\label{p2_3}
\end{align}
Thus, for any $i\in\b{S}_j$: $\theta c_j(\tau_j)-h(\theta)+\sum_{s\in\b{\Omega}_l,s\neq j}[\theta c_s(t_s)-h(\theta)]-r_l\geq\theta c_j(\tau_j)-h(\theta)+\sum_{s\in\b{\Omega}_i,s\neq j}[\theta c_s(t_s)-h(\theta)]-r_i$. Adding $\theta c_j(t_j)-\theta c_j(\tau_j)$ on both sides, we obtain, $\forall i\in\b{S}_j$:
\begin{align}
\sum_{s\in\b{\Omega}_l}[\theta c_s(t_s)-h(\theta)]-r_l\geq\sum_{s\in\b{\Omega}_i}[\theta c_s(t_s)-h(\theta)]-r_i.\label{p2_1}
\end{align}
On the other hand, for any $i\notin\b{S}_j$:
\begin{align}
&\sum_{s\in\b{\Omega}_i}[\theta c_s(t_s)-h(\theta)]-r_i\\
&\leq\theta c_j(\tau_j)-h(\theta)+\sum_{s\in\b{\Omega}_l,s\neq j}[\theta c_s(t_s)-h(\theta)]-r_l\\
&\aleq\sum_{s\in\b{\Omega}_l}[\theta c_s(t_s)-h(\theta)]-r_l,\label{p2_2}
\end{align}
where (a) follows from the fact that $c_j(\tau_j)\leq c_j(t_j)$. Combining inequalities \eqref{p2_1} and \eqref{p2_2}, we thus conclude $l=\arg\max_{i=1,...,m}\left\{\sum_{r\in\b{\Omega}_i}[\theta c_s(t_s)-h(\theta)]-r_i\right\}$. Moreover, we note that $\max_{i=1,...,m}\left\{\sum_{s\in\b{\Omega}_i}[\theta c_s(t_s)-h(\theta)]-r_i\right\}=\sum_{s\in\b{\Omega}_l}[\theta c_s(t_s)-h(\theta)]-r_l\geq\theta c_j(\tau_j)-h(\theta)+\sum_{s\in\b{\Omega}_l,s\neq j}[\theta c_s(t_s)-h(\theta)]-r_l>0$, where the second last step results from $c_j(\tau_j)\leq c_j(t_j)$ and the last step follows from \eqref{p2_3}. Hence, according to the construction of $\@{p}^*(\cdot)$, we have $p_l^*(\@{t})=1$, $p_i^*(\@{t})=0$, $\forall i\neq l$. So, $\sum_{i\in\b{S}_j}p_i^*(\@{t})=1=\sum_{i\in\b{S}_j}p_i^*(\tau_j,\@{t}_{-j})$.

\emph{Case 2: $\sum_{i\in\b{S}_j}p_i^*(\tau_j,\@{t}_{-j})=0$.} In such a case, since $\sum_{i\in\b{S}_j}p_i^*(\@{t})$ is either 0 or 1, we evidently have: $\sum_{i\in\b{S}_j}p_i^*(\tau_j,\@{t}_{-j})\leq\sum_{i\in\b{S}_j}p_i^*(\@{t})$.

Combining cases 1 and 2, we always have $\sum_{i\in\b{S}_j}p_i^*(\tau_j,\@{t}_{-j})\leq\sum_{i\in\b{S}_j}p_i^*(\@{t})$, for any $\@{t}_{-j}\in\b{T}_{-j}$. Thus, $\int_{\b{T}_{-j}}\left[\sum_{i\in\b{S}_j}p_i^*(\tau_j,\@{t}_{-j})\right]f_{-j}(\@{t}_{-j})d\@{t}_{-j}\leq\int_{\b{T}_{-j}}\left[\sum_{i\in\b{S}_j}p_i^*(\@{t})\right]f_{-j}(\@{t}_{-j})d\@{t}_{-j}$. Hence, $\@{p}^*(\cdot)$ satisfies statement (i) in Lemma \ref{lem_constraints} and it is optimal for problem \eqref{prob_p}.
\end{proof}

\begin{rem}
The optimal cache space allocation mechanism $\@{p}^*(\cdot)$ can be explained as follows. Suppose all users report their true types $\@{t}$, as guaranteed by the IC constraint. If the SP allocates all cache space to content $i$, i.e., $p_i(\@{t})=0$, $p_l(\@{t})=0$, for any $l\neq i$, then the social welfare comprised of the SP's profit and all users' utilities is $\sum_{j\in\b{\Omega}_i}(\theta t_j-h(\theta))-r_i$. If we replace the true valuation $t_j$ with a \emph{virtual} valuation $c_j(t_j)$, then the virtual social welfare becomes $\sum_{j\in\b{\Omega}_i}(\theta c_j(t_j)-h(\theta))-r_i$. Thus, according to Proposition \ref{p_opt}, the optimal allocation mechanism $\@{p}^*(\@{t})$ is to allocate all the cache space to the content $i$ that yields the greatest virtual social welfare if this greatest virtual social welfare is positive. Otherwise, if the greatest virtual social welfare is negative (i.e., the virtual social welfare of every content is negative), then the SP will not cache any content. Therefore, the optimal allocation mechanism is to maximize the virtual social welfare (including zero in which case the SP caches nothing). Note that this assertion can only facilitate our comprehension of the mechanism. It cannot be directly implemented since the virtual social welfare is an ex-post quantity, i.e., it depends on the realizations of types $\@{t}$ which are private information of users. In fact, the reason that we use virtual valuation $c_j(t_j)$ instead of the true valuation $t_j$ is just to incentivize users to report their private types truthfully.
\end{rem}

Propositions \ref{x_opt} and \ref{p_opt} together specify the optimal auction mechanism $\langle\@{p}^*(\cdot),\@{x}^*(\cdot)\rangle$ for problem \eqref{opt_auc_fix_theta}, as summarized in the following theorem.

\begin{thm}\label{sol_fix_theta}
Suppose the optimal auction problem \eqref{opt_auc_fix_theta} is regular. The optimal solution $\langle\@{p}^*(\cdot),\@{x}^*(\cdot)\rangle$ to problem \eqref{opt_auc_fix_theta} is given as follows. For each fixed $\@{t}\in\b{T}$, define $k=\arg\max_{i=1,...,m}\left\{\sum_{j\in\b{\Omega}_i}[\theta c_j(t_j)-h(\theta)]-r_i\right\}$. Then, the optimal cache space allocation mechanism $\@{p}^*(\@{t})$ at this particular $\@{t}$ is prescribed by equation \eqref{p_sol}. Furthermore, the optimal payment mechanism is given by:
\begin{align}
x_j^*(\@{t})=\theta t_j\sum_{i\in\b{S}_j}p_i^*(\@{t})-\theta\int_{\u{a}_j}^{t_j}\sum_{i\in\b{S}_j}p_i^*(\tau_j,\@{t}_{-j})d\tau_j,\nonumber\\
~\forall j\in\{1,...,n\},~\forall \@{t}\in\b{T}.
\end{align}
\end{thm}

According to Theorem \ref{sol_fix_theta}, to compute the optimal payment $x_j^*(\@{t})$, we need to calculate the integral $\int_{\u{a}_j}^{t_j}\sum_{i\in\b{S}_j}p_i^*(\tau_j,\@{t}_{-j})d\tau_j$ for arbitrarily given $\@{t}$, which is hard to implement directly. Thanks to the special structure of the optimal allocation mechanism $\@{p}^*(\cdot)$, we can compute this integral (and thus the payment $x_j^*(\@{t})$) efficiently for any given $\@{t},j$, as stated in the following.

\begin{prop}\label{prop_payment}
Suppose the optimal auction problem \eqref{opt_auc_fix_theta} is regular. For each given $\@{t}\in\b{T}$ and $j\in\{1,...,n\}$, we define a number $\beta_j$ and a function $\phi_j:\b{T}_j\mapsto\b{R}$ as follows (their dependence on $\@{t}$ is suppressed):
\begin{align}
&\beta_j\delequal\max\left\{0,\max_{i\notin\b{S}_j}\left\{\sum_{s\in\b{\Omega}_i}[\theta c_s(t_s)-h(\theta)]-r_i\right\}\right\},\label{beta_def}\\
&\phi_j(\tau_j)\delequal\theta c_j(\tau_j)-h(\theta)\nonumber\\
&~~~~~~~~+\max_{i\in\b{S}_j}\left\{\sum_{s\in\b{\Omega}_i,s\neq j}[\theta c_s(t_s)-h(\theta)]-r_i\right\},~\forall\tau_j\in\b{T}_j.
\end{align}
Then, we have:
\begin{align}\label{integral}
\int_{\u{a}_j}^{t_j}\sum_{i\in\b{S}_j}p_i^*(\tau_j,\@{t}_{-j})d\tau_j=
\begin{cases}
t_j-\u{a}_j,~\text{if}~\phi_j(\u{a}_j)\geq\beta_j,\\
0,~~~~~~~\text{if}~\phi_j(t_j)<\beta_j,\\
t_j-\xi_j,~\text{if}~\phi_j(\u{a}_j)<\beta_j\leq\phi_j(t_j).
\end{cases}
\end{align}
In the last case of \eqref{integral}, we have defined:
\begin{align}
\xi_j\delequal &c_j^{-1}\Bigg(\frac{1}{\theta}\Bigg[\beta_j+h(\theta)-\max_{i\in\b{S}_j}\Bigg\{\sum_{s\in\b{\Omega}_i,s\neq j}[\theta c_s(t_s)-h(\theta)]\nonumber\\
&-r_i\Bigg\}\Bigg]\Bigg),\label{xi_def}
\end{align}
where $c_j^{-1}(\cdot)$ means the inverse function of $c_j(\cdot)$\footnote{A sufficient condition for the existence of the inverse function $c_j^{-1}(\cdot)$ is that $c_j(\cdot)$ is \emph{strictly} increasing and continuous. Otherwise, one can replace $c_j^{-1}(z_j)$ with $\inf\{t_j'\in\b{T}_j|c_j(t_j')\geq z_j\}$ since the regularity assumption has already guaranteed that $c_j(\cdot)$ is (weakly) increasing.}.
\end{prop}
\begin{proof}
Fix arbitrary $\@{t}\in\b{T}$ and $j\in\{1,...,n\}$. Define:
\begin{align}
l&\delequal\arg\max_{i\in\b{S}_j}\left\{\sum_{s\in\b{\Omega}_i}[\theta c_s(t_s)-h(\theta)]-r_i\right\}\\
&=\arg\max_{i\in\b{S}_j}\left\{\sum_{s\in\b{\Omega}_i,s\neq j}[\theta c_s(t_s)-h(\theta)]-r_i\right\},\label{l_def}
\end{align}
in which we suppress the dependence of $l$ on $\@{t}$ and $j$ to avoid cluttered notations. Suppose $l$ is the unique maxima for \eqref{l_def}, which happens with probability 1. Thus, for any $\tau_j\in\b{T}_j$, $i\in\b{S}_j,i\neq l$, we have:
\begin{align}
\phi_j(\tau_j)&=\theta c_j(\tau_j)-h(\theta)+\sum_{s\in\b{\Omega}_l,s\neq j}[\theta c_s(t_s)-h(\theta)]-r_l\label{p3_1}\\
&>\theta c_j(\tau_j)-h(\theta)+\sum_{s\in\b{\Omega}_i,s\neq j}[\theta c_s(t_s)-h(\theta)]-r_i.\label{p3_2}
\end{align}
So, from the construction of $\@{p}^*(\cdot)$ in Proposition \ref{p_opt}, we know that $p_i^*(\tau_j,\@{t}_{-j})=0$. Hence, we have:
\begin{align}
\int_{\u{a}_j}^{t_j}\sum_{i\in\b{S}_j}p_i^*(\tau_j,\@{t}_{-j})d\tau_j=\int_{\u{a}_j}^{t_j}p_l^*(\tau_j,\@{t}_{-j})d\tau_j.\label{p3_3}
\end{align}
Note that $\phi_j(\cdot)$ is an increasing function on $\b{T}_j$. We thus distinguish three cases.

\emph{Case 1: $\phi_j(\u{a}_j)\geq\beta_j$.} In such a case, $\phi_j(\tau_j)\geq\beta_j$, for any $\tau_j\in[\u{a}_j,t_j]$. From the definition of $\beta_j$ in \eqref{beta_def}, equations \eqref{p3_1}, \eqref{p3_2} and the construction rule of $\@{p}^*(\cdot)$, one can easily see that $p_l^*(\tau_j,\@{t}_{-j})=1$, for any $\tau_j\in[\u{a}_j,t_j]$. Thus, $\int_{\u{a}_j}^{t_j}p_l^*(\tau_j,\@{t}_{-j})d\tau_j=t_j-\u{a}_j$.

\emph{Case 2: $\phi_j(t_j)<\beta_j$.} In such a case, we have $\phi_j(\tau_j)<\beta_j$, for any $\tau_j\in[\u{a}_j,t_j]$. Analogously, we can assert that $p_l^*(\tau_j,\@{t}_{-j})=0$, for any $\tau_j\in[\u{a}_j,t_j]$. So, $\int_{\u{a}_j}^{t_j}p_l^*(\tau_j,\@{t}_{-j})d\tau_j=0$.

\emph{Case 3: $\phi_j(\u{a}_j)<\beta_j\leq\phi_j(t_j)$.} In such a case, from the definition of $\xi_j$ in \eqref{xi_def}, we know that it is the unique solution of the equation $\phi_j(\xi_j)=\beta_j$ over the interval $[\u{a}_j,t_j]$. Hence, for $\tau_j\in[\u{a}_j,\xi_j)$, we have $\phi_j(\tau_j)<\beta_j$ and thus $p_l^*(\tau_j,\@{t}_{-j})=0$. On the other hand, for $\tau_j\in[\xi_j,t_j]$, we have $\phi_j(\tau_j)\geq\beta_j$ and thus $p_l^*(\tau_j,\@{t}_{-j})=1$. Combining these two situations, we derive $\int_{\u{a}_j}^{t_j}p_l^*(\tau_j,\@{t}_{-j})d\tau_j=t_j-\xi_j$.

Combining the results in the three cases and noting relation \eqref{p3_3}, we conclude the proposition.
\end{proof}

We note that the R.H.S. of \eqref{integral} can be calculated easily for arbitrarily given $\@{t}\in\b{T}$ and $j\in\{1,...,n\}$.  Therefore, Proposition \ref{prop_payment} provides us a simple way of computing the optimal payment $x_j^*(\@{t})$. Moreover, making use of Proposition \ref{prop_payment}, we can show the following intuitively reasonable proposition of the optimal mechanism $\langle\@{p}^*(\cdot),\@{x}^*(\cdot)\rangle$.

\begin{prop}\label{property}
Suppose the optimal auction problem \eqref{opt_auc_fix_theta} is regular. For any $\@{t}\in\b{T}$ and any $j\in\{1,...,n\}$, if $\sum_{i\in\b{S}_j}p_i^*(\@{t})=0$, then $x_j^*(\@{t})=0$.
\end{prop}
\begin{proof}
Suppose $\sum_{i\in\b{S}_j}p_i^*(\@{t})=0$. We will show by contradiction that $\phi_j(t_j)<\beta_j$. Otherwise, if $\phi_j(t_j)\geq\beta_j$, then we have 
\begin{align}
\theta c_j(t_j)-h(\theta)+\sum_{s\in\b{\Omega}_l,s\neq j}[\theta c_s(t_s)-h(\theta)]-r_l\geq\beta_j,\label{property_2}
\end{align}
where we define
$$l\delequal\arg\max_{i\in\b{S}_j}\left\{\sum_{s\in\b{\Omega}_i,s\neq j}[\theta c_s(t_s)-h(\theta)]-r_i\right\}.$$
From the definition of $\beta_j$ in \eqref{beta_def}, we have for any $i\notin\b{S}_j$:
\begin{align}
\sum_{s\in\b{\Omega}_l}[\theta c_s(t_s)-h(\theta)]-r_l\geq\sum_{s\in\b{\Omega}_i}[\theta c_s(t_s)-h(\theta)]-r_i.\label{property_1}
\end{align}
From the definition of $l$, we can see that \eqref{property_1} holds for $i\in\b{S}_j$ as well. Hence, we have $l=\arg\max_{i=1,...,m}\left\{\sum_{s\in\b{\Omega}_i}[\theta c_s(t_s)-h(\theta)]-r_i\right\}$. From \eqref{property_2} and $\beta_j\geq0$, we further know that $\max_{i=1,...,m}\left\{\sum_{s\in\b{\Omega}_i}[\theta c_s(t_s)-h(\theta)]-r_i\right\}\geq0$. According to the construction of $\@{p}^*(\@{t})$, we thus have $p_l^*(\@{t})=1$ and $p_i^*(\@{t})=0$ for any $i\neq l$. Noting that $l\in\b{S}_j$, we get $\sum_{i\in\b{S}_j}p_i^*(\@{t})=1$, which is a contradiction. So, we must have $\phi_j(t_j)<\beta_j$, which implies $\int_{\u{a}_j}^{t_j}\sum_{i\in\b{S}_j}p_i^*(\tau_j,\@{t}_{-j})d\tau_j=0$ according to Proposition \ref{prop_payment}. Hence, according to the construction of $\@{x}^*(\cdot)$ in Theorem \ref{sol_fix_theta}, we obtain $x_j^*(\@{t})=0$.
\end{proof}
\begin{rem}
Proposition \ref{property} asserts that, in the optimal mechanism, a user does not pay anything if no content of her interest is cached by the SP, as expected. In such a case, her ex-post utility is also zero, i.e., $u(\@{t},\@{t})=0$ (c.f. Equation \eqref{u}). Furthermore, we note that for each $\@{t}\in\b{T}$, $\@{p}^*(\@{t})$ is either $\@{0}$ or $\@{e}_k$ for some $k\in\{1,...,m\}$, where $\@{e}_k=[0,...,0,1,0,...,0]^\s{T}$ (the sole 1 takes place at the $k$-th entry). Thus, for the condition in Proposition \ref{property} to hold, we distinguish the following two cases for each $\@{t}\in\b{T}$. If there exists $k$ such that $\@{p}^*(\@{t})=\@{e}_k$, then $\sum_{i\in\b{S}_j}p_i^*(\@{t})=0$, for any $j\notin\b{\Omega}_k$. If $\@{p}^*(\@{t})=\@{0}$, then obviously $\sum_{i\in\b{S}_j}p_i^*(\@{t})=0$ for any $j=1,...,n$. The assertion in Proposition \ref{property} will be confirmed empirically through numerical experiments in Section \Rmnum{4}.
\end{rem}

\subsection{Optimial Determination of Content Delivery Quality $\theta$}

The optimization problem \eqref{opt_auc_fix_theta} is for fixed delivery quality $\theta>0$. In this section, we let $\theta$ vary and determine the optimal $\theta$ to further maximize the ER, i.e., solving the optimization problem \eqref{opt_auc}. By using the optimal auction mechanism $\langle\@{p}^*(\cdot),\@{x}^*(\cdot)\rangle$ specified in Theorem \ref{sol_fix_theta}, the corresponding optimal value of the ER becomes a function of $\theta$ as follows:
\begin{align}
&\texttt{ER}^*(\theta)\nonumber\\
&=\int_\b{T}\left\{\sum_{i=1}^m\left[\sum_{j\in\b{\Omega}_i}(\theta c_j(t_j)-h(\theta))-r_i\right]p_i^*(\@{t})\right\}f(\@{t})d\@{t},\label{er_theta}
\end{align}
which is difficult to evaluate in general based on the construction of $\@{p}^*(\cdot)$ in Proposition \ref{p_opt}. For analysis tractability, we next focus on the case meeting the following assumptions.
\begin{assump}\label{assump_large_user}
The number of users $n$ is large.
\end{assump}
\begin{assump}\label{assump_iid}
The type distributions of all users are the same, i.e., $f_j(\cdot)=f(\cdot)$, $F_j(\cdot)=F(\cdot)$, $\u{a}_j=\u{a}>0$, $\o{a}_j=\o{a}$ for any $j\in\{1,...,n\}$. Moreover, this common distribution satisfies the regularity condition in Definition \ref{regularity}.
\end{assump}
\begin{assump}\label{assump_formation}
For each content $i\in\{1,...,m\}$, each user $j\in\{1,...,n\}$ is included into $\b{\Omega}_i$, i.e., user $j$ is interested in content $i$, with probability $q_i\in[0,1]$ independently.
\end{assump}
\begin{assump}\label{assump_cost_func}
The cost function $h(\cdot)$ satisfies the following properties. $h(0)=0$; $h$ is convex; $h'(0)=0$; and $\lim_{\theta\rightarrow+\infty}h'(\theta)=\infty$.
\end{assump}

Assumption \ref{assump_large_user} is reasonable since the user density has increased dramastically in the recent decade with the advancement of mobile networks and devices. Assumption \ref{assump_iid} is a homogeneity hypothesis commonly used in the analysis of large-scale systems. Assumption \ref{assump_formation} is widely adopted in the popularity modeling of content centric networks, where $q_i$ characterizes the popularity of content $i$, e.g., the Zipf distribution of content popularity \cite{golrezaei2014scaling,ji2015throughput,jeon2017wireless}. Assumption \ref{assump_cost_func} collects common properties of cost functions in resource allocation \cite{georgiadis2006resource}. Under these assumptions, the optimal $\theta$ for problem \eqref{opt_auc}, or equivalently the $\theta$ that maximizes $\texttt{ER}^*(\theta)$ in \eqref{er_theta}, can be computed as in the following theorem.

\begin{thm}\label{thm_opt_theta}
Suppose that Assumptions \ref{assump_large_user}-\ref{assump_cost_func} hold. Then, the optimal $\theta$ for problem \eqref{opt_auc}, maximizing $\texttt{ER}^*(\theta)$ in \eqref{er_theta}, is given by $\theta^*=(h')^{-1}(\u{a})$, where $(h')^{-1}$ means the inverse function of $h'$, the derivative of $h(\cdot)$.
\end{thm}
\begin{proof}
From Assumption \ref{assump_iid}, we note that $c_j(\cdot)=c(\cdot)$, $\forall j=1,...,n$. For any $\@{t}\in\b{T}$ and any $i\in\{1,...,m\}$, we have:
\begin{align}
&\frac{1}{n}\left\{\sum_{j\in\b{\Omega}_i}[\theta c(t_j)-h(\theta)]-r_i\right\}\nonumber\\
&=\frac{|\b{\Omega}_i|}{n}\left[\frac{\theta}{|\b{\Omega}_i|}\left(\sum_{j\in\b{\Omega}_i}c(t_j)\right)-h(\theta)\right]-\frac{r_i}{n}.\label{p4_1}
\end{align}
By the strong law of large numbers (SLLN) \cite{gallager2012discrete}, we know that $\frac{|\b{\Omega}_i|}{n}=\frac{1}{n}\sum_{j=1}^n\b{1}(j\in\b{\Omega}_i)\xrightarrow{\text{a.s.}}\b{E}[\b{1}(1\in\b{\Omega}_i)]=q_i$, where $\b{1}(\cdot)$ is the indicator function. Moreover, by SLLN, we have $\frac{1}{|\b{\Omega}_i|}\sum_{j\in\b{\Omega}_i}c(t_j)\xrightarrow{\text{a.s.}}\b{E}[c(t_1)]=\int_{\b{T}_1}\left[t_1-\frac{1-F(t_1)}{f(t_1)}\right]f(t_1)dt_1=\u{a}$, where we use integration by parts in the last step. Substituting these limits into \eqref{p4_1}, we obtain $\frac{1}{n}\left\{\sum_{j\in\b{\Omega}_i}[\theta c(t_j)-h(\theta)]-r_i\right\}\xrightarrow{\text{a.s.}} q_i[\theta\u{a}-h(\theta)]$. Denote $k=\arg\max_{i=1,...,m}q_i$. According to the construction of $\@{p}^*(\cdot)$ in Proposition \ref{p_opt}, we distinguish two cases. If $\theta\u{a}-h(\theta)>0$, then $p_k^*(\@{t})=1$, $p_i^*(\@{t})=0,\forall i\neq k$, $\forall\@{t}\in\b{T}$. In such a case, for any $\@{t}\in\b{T}$, $\frac{1}{n}\sum_{i=1}^m\left\{\sum_{j\in\b{\Omega}_i}[\theta c_j(t_j)-h(\theta)]-r_i\right\}p_i^*(\@{t})=\frac{1}{n}\left\{\sum_{j\in\b{\Omega}_k}[\theta c_j(t_j)-h(\theta)]-r_k\right\}\xrightarrow{\text{a.s.}}q_k[\theta\u{a}-h(\theta)]$ and $\texttt{ER}^*(\theta)\approx nq_k[\theta\u{a}-h(\theta)]$. Otherwise, if $\theta\u{a}-h(\theta)\leq0$, then $p_i^*(\@{t})=0$, $\forall i=1,..,m,\forall\@{t}\in\b{T}$ and $\texttt{ER}^*(\theta)=0$. According to Assumption \ref{assump_cost_func}, we know that $[\theta\u{a}-h(\theta)]\big|_{\theta=0}=0$ and $\frac{d}{d\theta}[\theta\u{a}-h(\theta)]\big|_{\theta=0}=\u{a}>0$. Thus, $\theta\u{a}-h(\theta)>0$ for small enough $\theta>0$. As such, the optimal $\theta^*$ that maximizes $\texttt{ER}^*(\theta)$ is $\theta^*=\arg\max_{\theta>0}\{\theta\u{a}-h(\theta)\}=(h')^{-1}(\u{a})$, where we make use of the convexity of $h(\cdot)$ in Assumption \ref{assump_cost_func}.
\end{proof}

\section{Numerical Results}

In this section, extensive numerical experiments are carried out to evaluate the performance of the proposed optimal auction mechanism for  content caching. In particular, the impact of various model parameters is studied empirically to get engineering insights into the content caching problem. All results involving expectations are average over $10^4$ independent trials.

\begin{figure*}
\renewcommand\figurename{\small Fig.}
\centering \vspace*{8pt} \setlength{\baselineskip}{10pt}
\subfigure[Uniform distribution: expected results]{
\includegraphics[scale = 0.21]{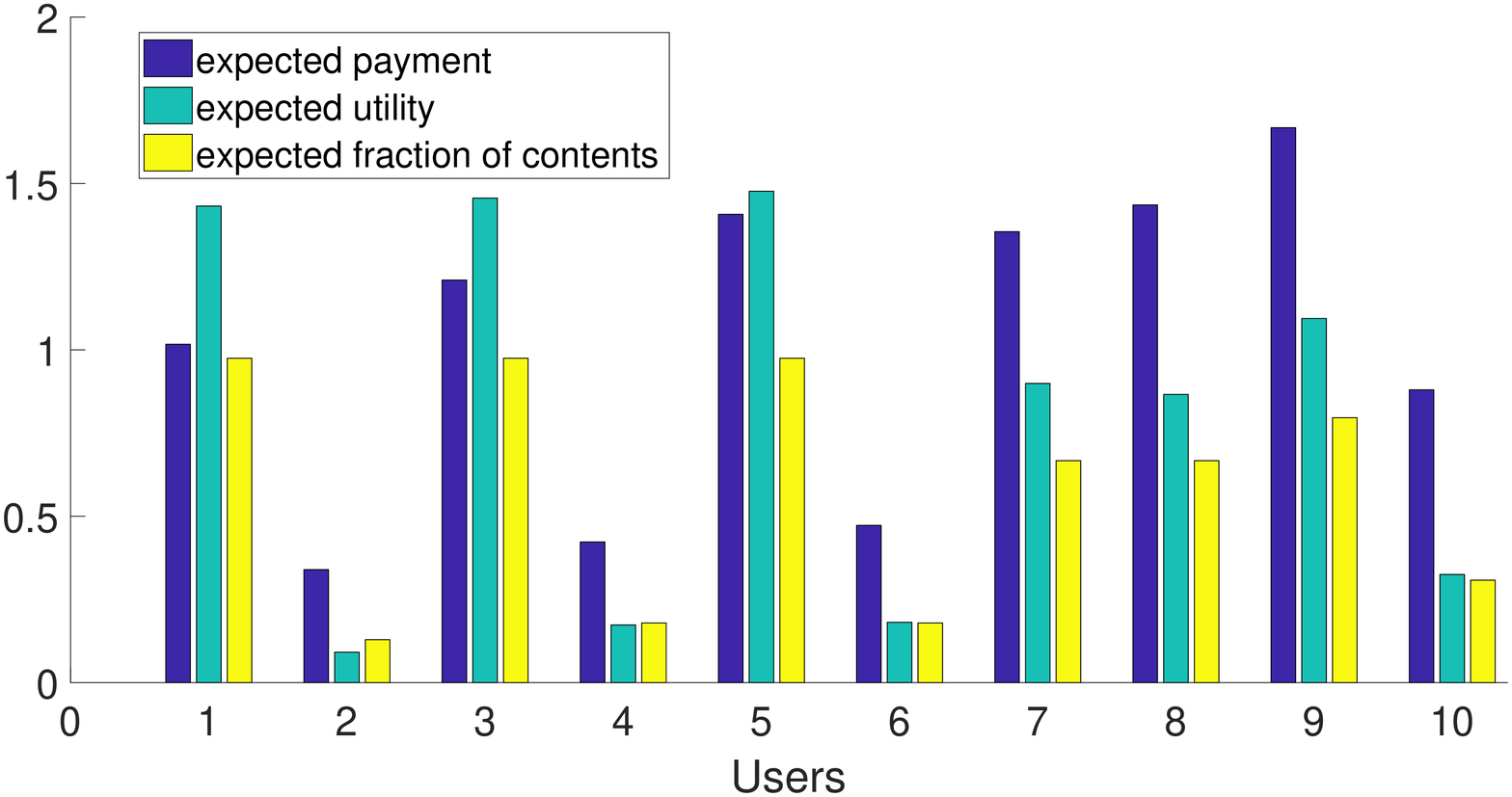}}
\subfigure[Exponential distribution: expected results]{
\includegraphics[scale = 0.21]{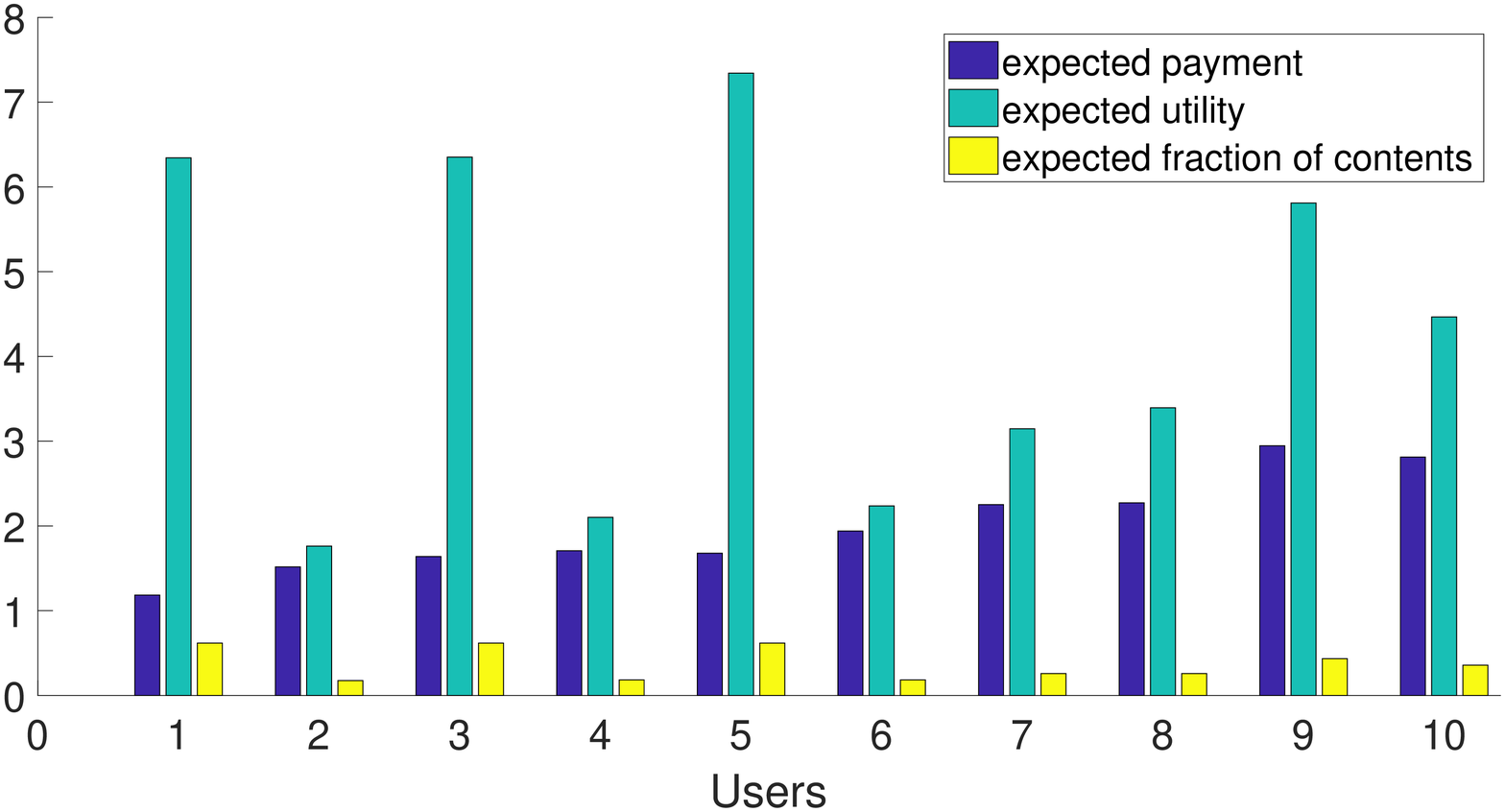}}
\subfigure[Uniform distribution: one realization]{
\includegraphics[scale = 0.21]{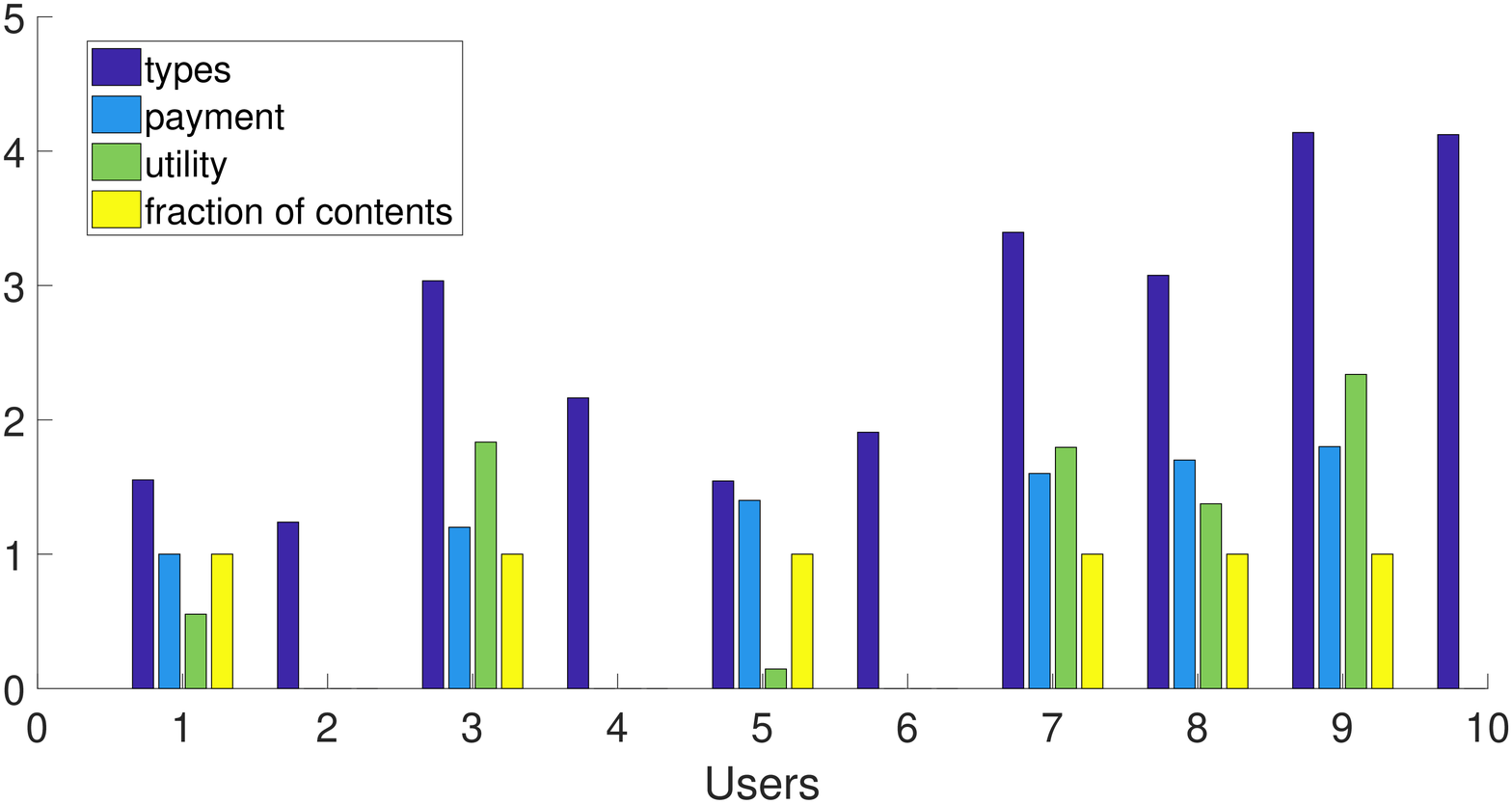}}
\subfigure[Exponential distribution: one realization]{
\includegraphics[scale = 0.21]{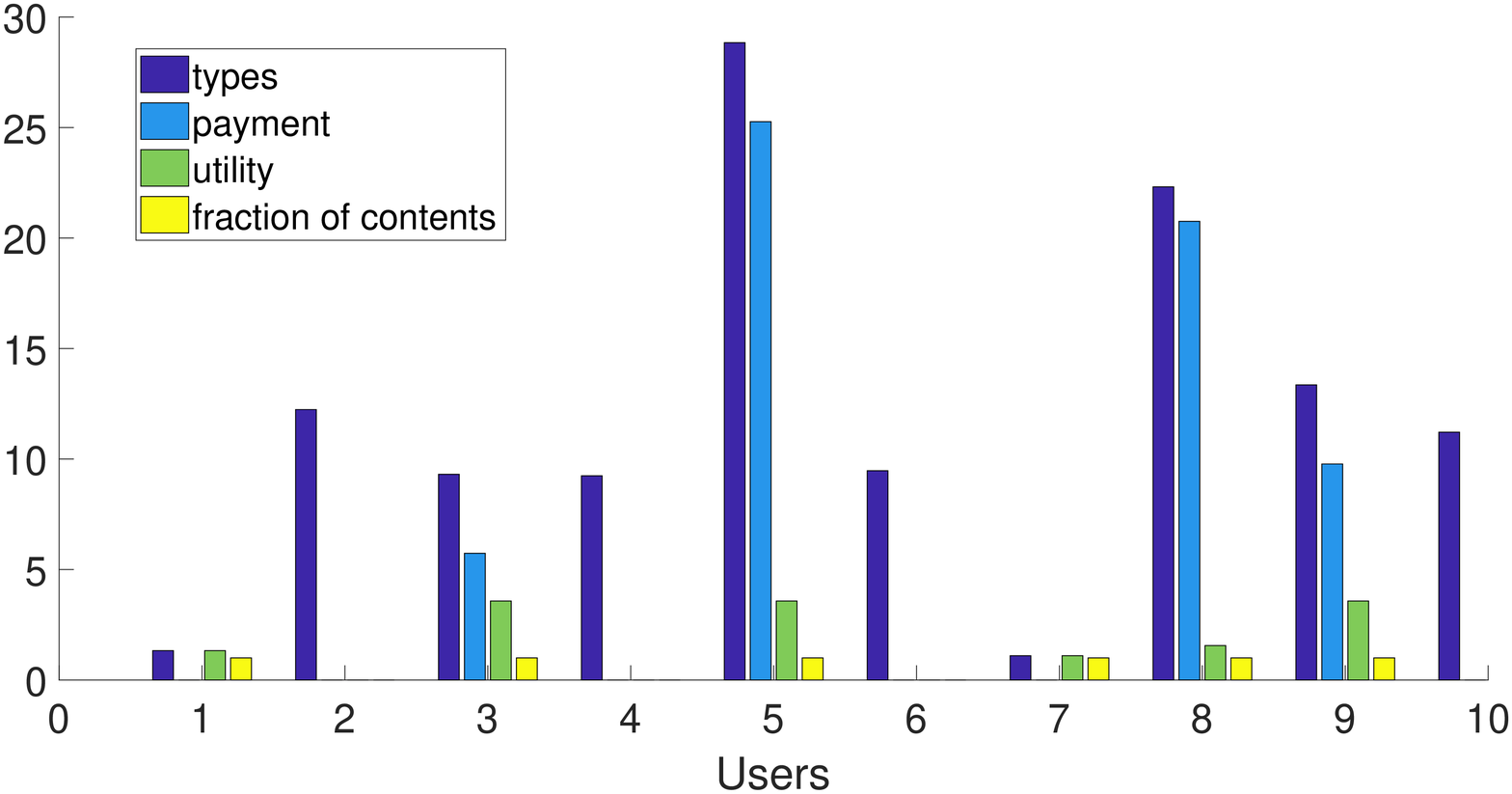}}
\caption{The payment, utility and fraction of contents of each user}
\label{fig_user}
\end{figure*}

Consider a content caching system (c.f. Fig. \ref{model}) with $m=3$ CPs and $n=10$ users. Each content interests 6 randomly chosen users as follows: $\b{\Omega}_1=\{1,3,4,5,6,10\}$, $\b{\Omega}_2=\{1,3,5,7,8,9\}$ and $\b{\Omega}_3=\{1,2,3,5,9,10\}$. The content delivery cost function is $h(\theta)=\alpha\theta^2$, where $\alpha=0.1$ for now. We set the content delivery quality to be $\theta=1$ temporarily. The acquisition costs of the 3 contents are set to be $r_1=4.2036$, $r_2=1.2714$, $r_3=4.0714$, which are chosen randomly. In the following simulations, we consider two user type distributions, namely uniform distribution and exponential distribution, both of which satisfy the regularity condition in Definition \ref{regularity}. For uniformly distributed user types, we set the the lower and upper bounds of user type supports to be $\u{a}_j=1+0.1(j-1)$ and $\o{a}_j=4+0.1(j-1)$, $j=1,...,n$, respectively. In addition, for exponentially distributed user types, i.e., $f_j(t_j)=\lambda_je^{-\lambda_jt_j}$ for $t_j\geq0$, we set the distribution parameters to be $\lambda_j=\frac{1}{10+0.4(j-1)}$, $j=1,...,n$, so that the expected types of users form an arithmetic progression. For either uniformly distributed types or exponentially distributed types, the expected types of users increase with user index.

We first study the expected payment $\b{E}[x_j^*(\@{t})]$, the expected utility $\b{E}[u_j(\@{t},\@{t})]=\b{E}\left[\left(\sum_{i\in\b{S}_j}p_i^*(\@{t})\right)\theta t_j\right]-\b{E}[x_j^*(\@{t})]$, and the expected fraction of contents $\b{E}\left[\sum_{i\in\b{S}_j}p_i^*(\@{t})\right]$ of each user $j$. The results for uniformly distributed user types and exponentially distributed user types are shown in Fig. \ref{fig_user}-(a) and Fig. \ref{fig_user}-(b), respectively. From Fig. \ref{fig_user}-(a), we observe that users with more expected fraction of contents generally tend to have higher expected utility and payment. Furthermore, among users with similar expected fraction of contents (e.g., users 1, 3, 5), those with higher expected types pay more, i.e., the SP extracts more profits from those users with high valuations than those with low valuations. These phenomena highlight the fairness of the proposed optimal auction mechanism. Analogous observations can be drawn from Fig. \ref{fig_user}-(b) for exponentially distributed user types. In addition, the expected cache space allocations $\b{E}[\@{p}^*(\@{t})]$ are $[0.17,0.657,0.146]^\s{T}$ and $[0.183,0.259,0.175]^\s{T}$ for uniformly distributed types and exponentially distributed types, respectively. We remark that the expected idle cache space $\b{E}[1-\@{1}^\s{T}\@{p}^*(\@{t})]$ at the SP for exponentially distributed types is greater than that for the uniformly distributed types. The reason is that the exponential distribution often realizes low types (the PDF is a decreasing function), which may lead to the decision of no caching at the SP. Besides, we note that the expected caching space allocated to content 2 is the largest among the three contents, in accordance with the large valuations of users in $\b{\Omega}_2$ (the user indices in $\b{\Omega}_2$ is the largest among the three contents). Moreover, we study the payment, utility and fraction of contents for users in one typical realization of types $\@{t}$. The results are shown in Fig. \ref{fig_user}-(c) and Fig. \ref{fig_user}-(d) for uniformly distributed types and exponentially distributed types, respectively. In either case, the cache space allocation is $\@{p}^*(\@{t})=[0,1,0]^\s{T}$, i.e., the entire cache space is allocated to content 2. We observe that payments are nonzero only for those users whose fraction of contents are positive, i.e., those in $\b{\Omega}_2$ in this case, as guaranteed by Proposition \ref{property}. Additionally, for users in $\b{\Omega}_2$ with similar type realizations (e.g., users 3 and 8 in Fig. \ref{fig_user}-(c)), those with higher \emph{expected} types pay more. The reason is that, to incentivize users with high expected types to report truthfully (as requested by the IC constraint), the SP will charge them relatively high payments even if their reports are low. This is to deter these users from reporting low types while possessing high true types, which are very probable since their expected types are high.

\begin{figure*}
\renewcommand\figurename{\small Fig.}
\centering \vspace*{8pt} \setlength{\baselineskip}{10pt}
\subfigure[ER of SP versus expected type for uniformly distributed types]{
\includegraphics[scale = 0.14]{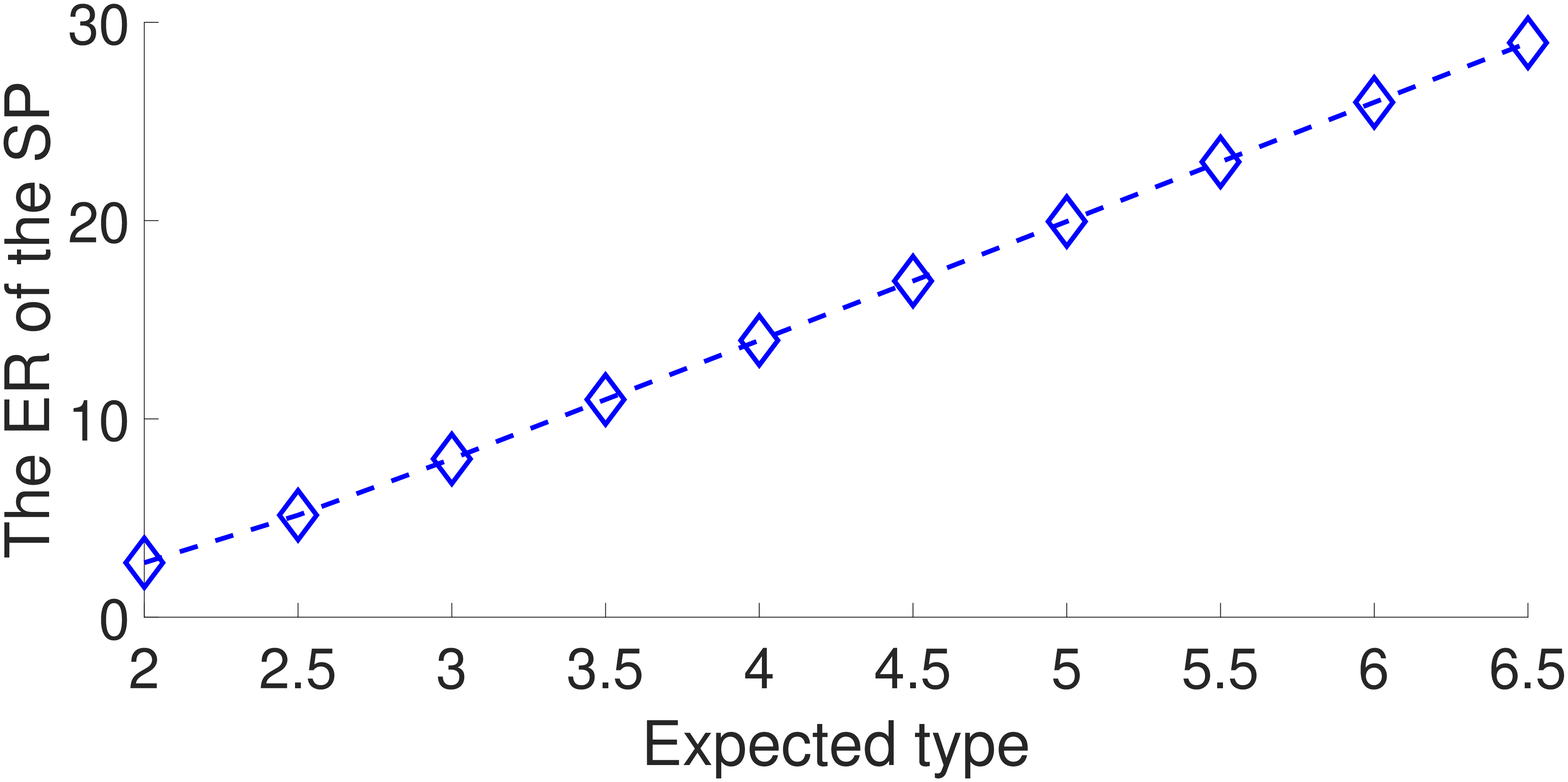}}
\subfigure[ER of SP versus length of support for uniformly distributed types]{
\includegraphics[scale = 0.14]{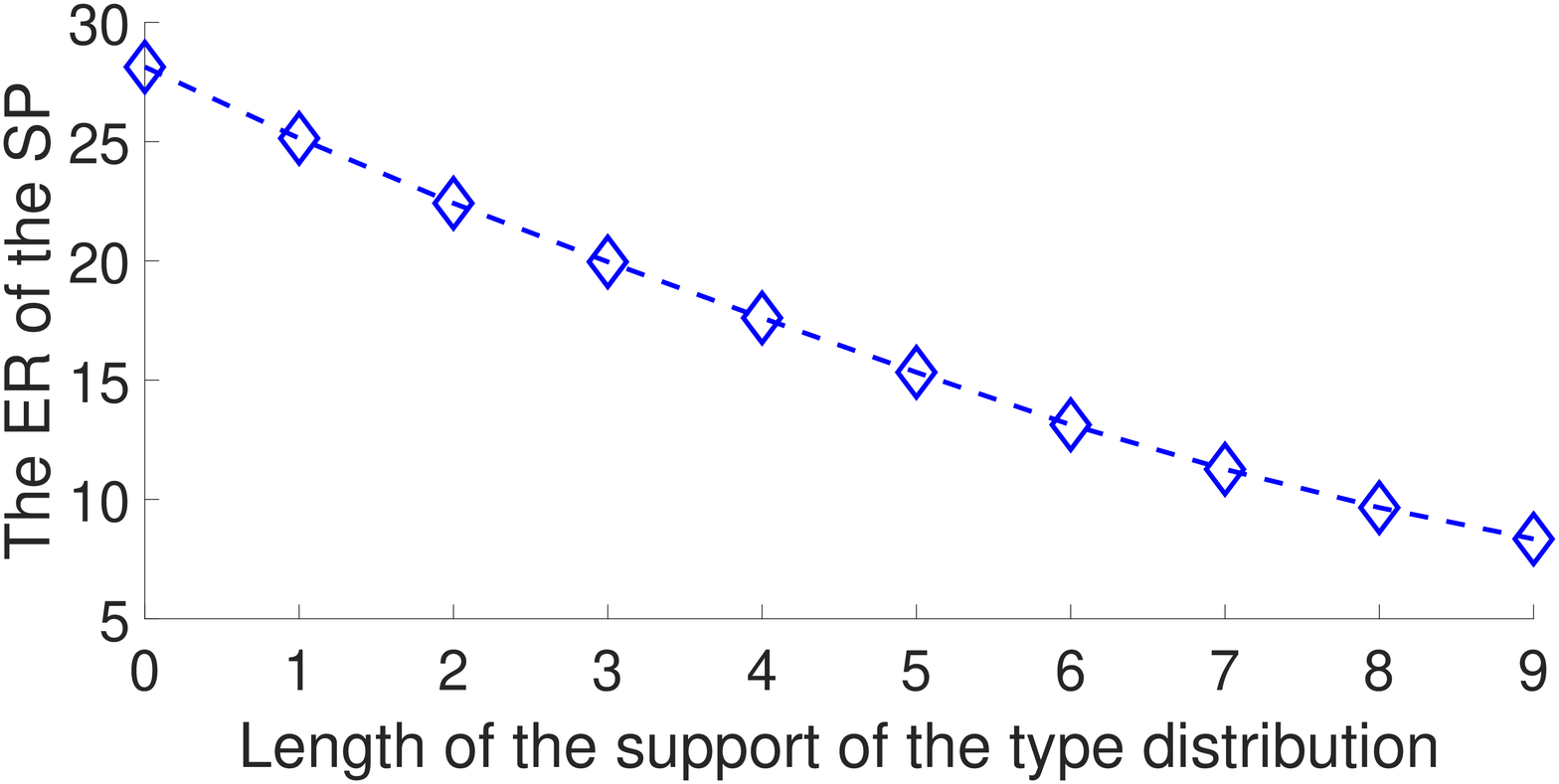}}
\subfigure[ER of SP versus $\lambda$ for exponentially distributed types]{
\includegraphics[scale = 0.14]{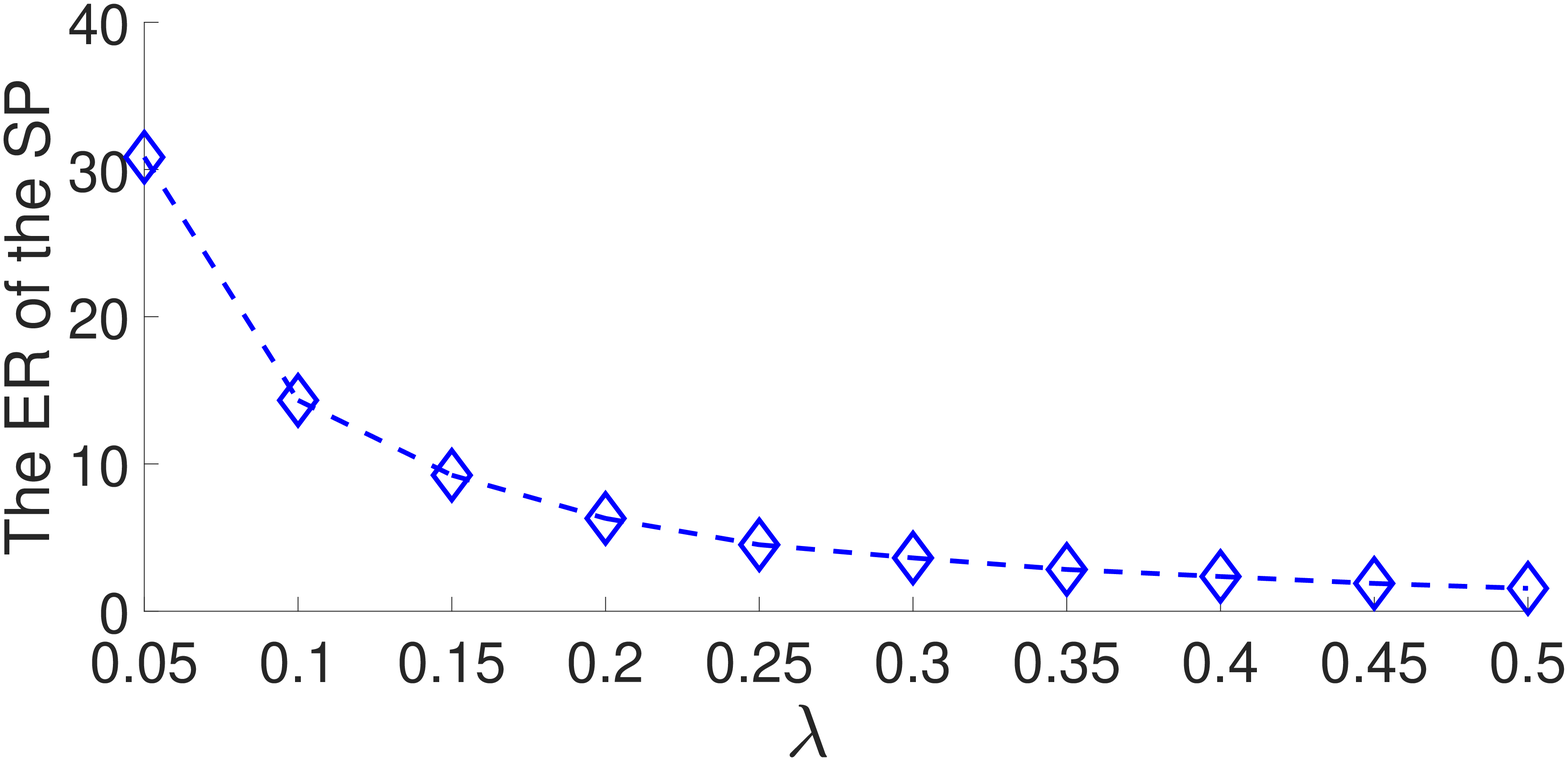}}
\subfigure[Expected average utility of users versus expected type for uniformly distributed types]{
\includegraphics[scale = 0.14]{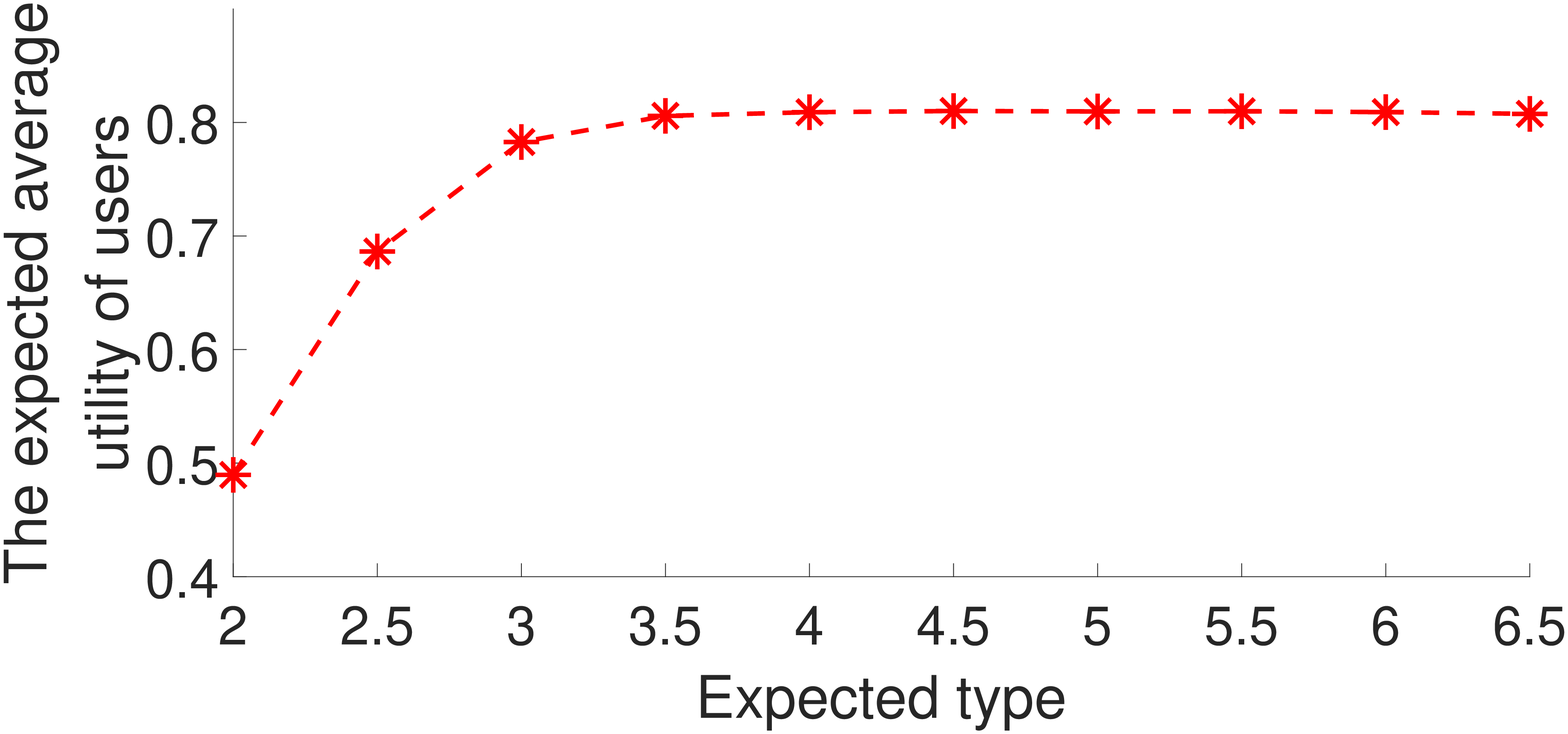}}
\subfigure[Expected average utility of users versus length of support for uniformly distributed types]{
\includegraphics[scale = 0.14]{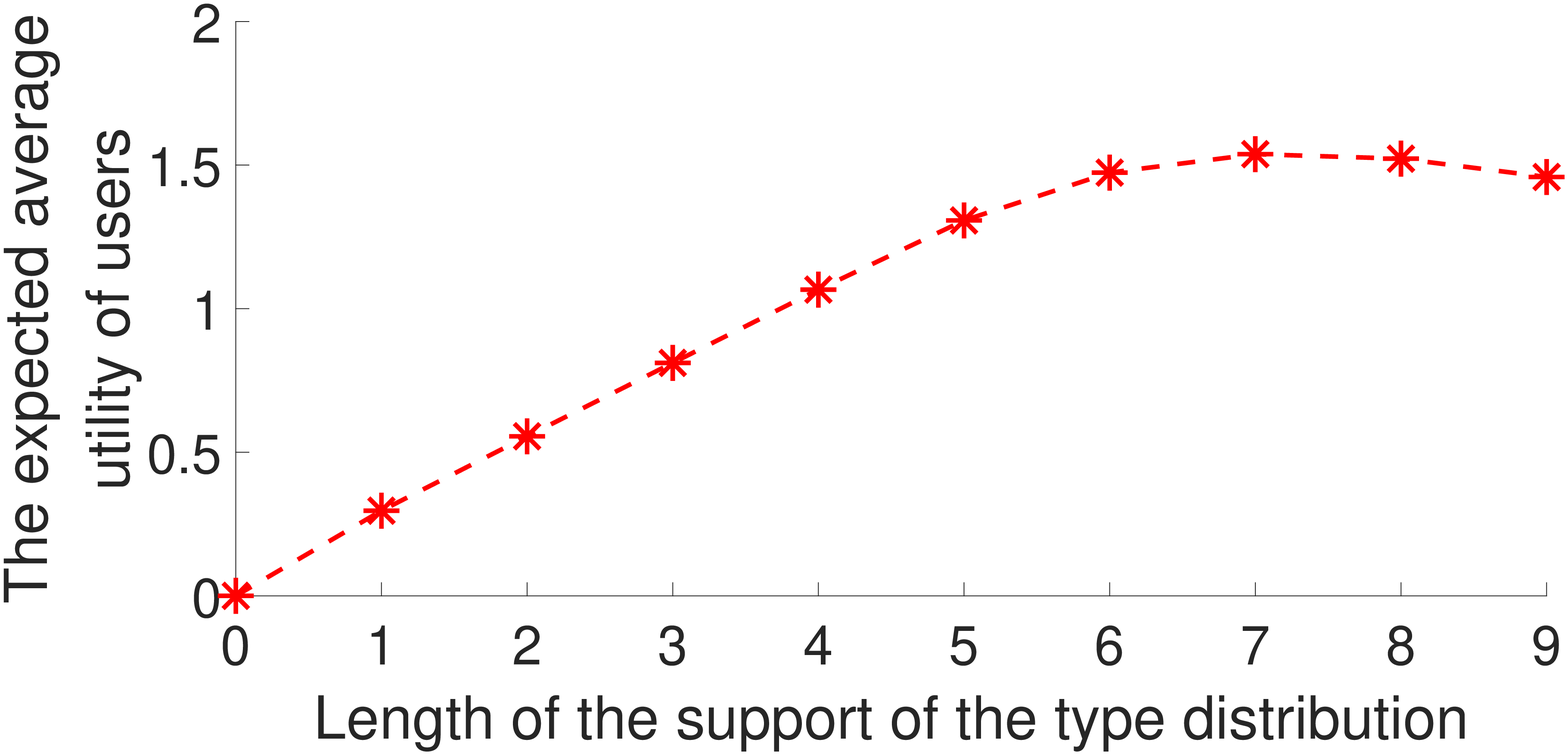}}
\subfigure[Expected average utility of users versus $\lambda$ for exponentially distributed types]{
\includegraphics[scale = 0.14]{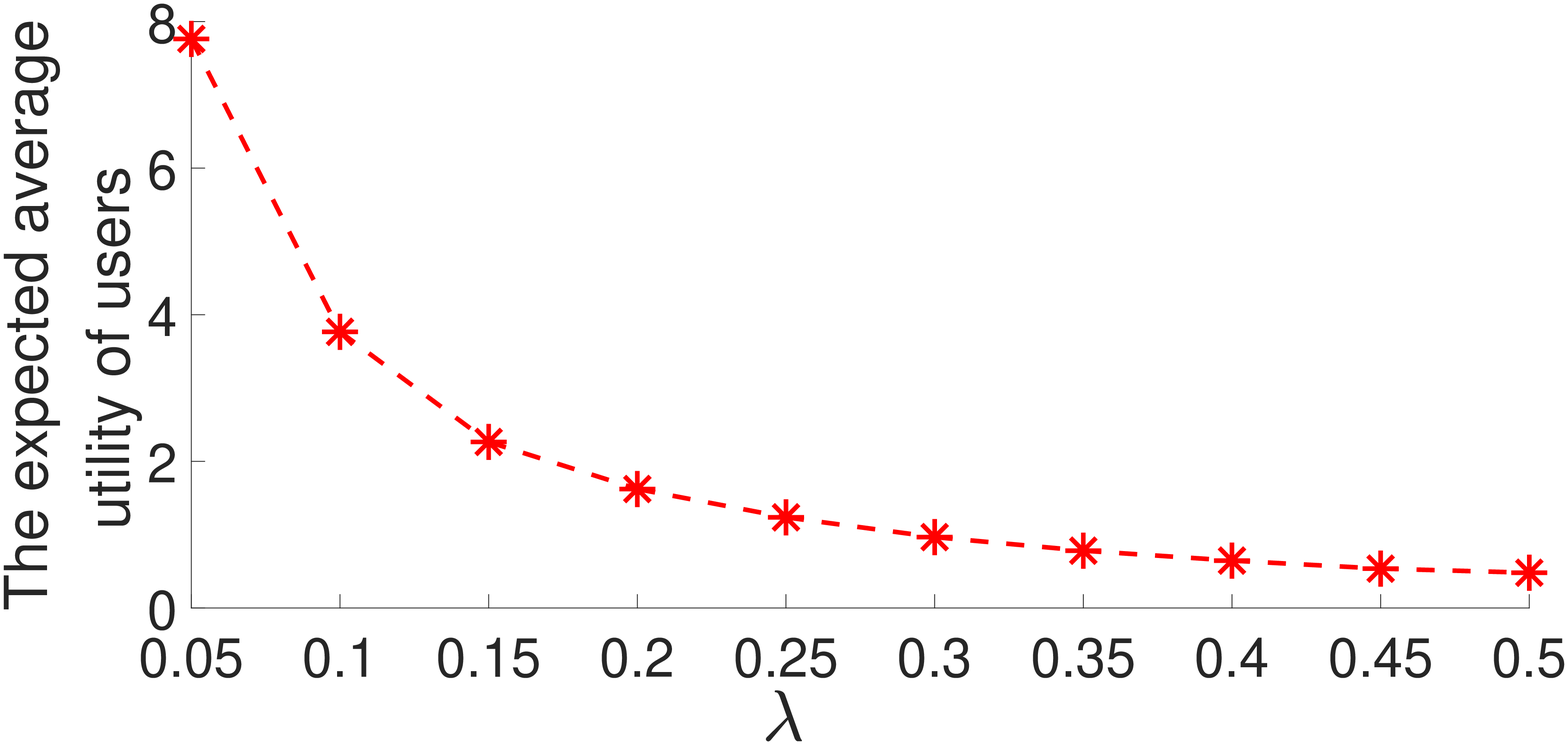}}
\caption{Impact of type distributions on the ER of the SP and the expected average utility of users}
\label{ErU}
\end{figure*}

Next, we investigate the impact of type distributions on the expected revenue (ER) of the SP and the expected average utility of users. Suppose user types are i.i.d. uniformly distributed with $\u{a}_j=\u{a}$, $\o{a}_j=\o{a}$ for any $j$. First, we remain the length of the support of the type distributions, i.e., $\o{a}-\u{a}$, to be fixed at 3 and let the expected type, i.e., $\frac{1}{2}(\o{a}+\u{a})$, vary from 2 to 6.5. In other words, we set $(\u{a},\o{a})$ to be $(0.5,3.5),(1,4),(1.5,4.5),...,(5,8)$. In such a way, the level of uncertainty of types remains unchanged while the expected types increase. The ER of the SP and the expected average utility of users are shown in Fig. \ref{ErU}-(a) and Fig. \ref{ErU}-(d), respectively. We observe that the ER of the SP increases with the expected valuation of users. This is reasonable because users with higher valuations tend to pay more to the SP. In addition, as the expected valuation of users increases, the expected average utility of users first increases and then saturates. This suggests that, when users' valuations are high enough, the SP may reap all the excess valuations from users in the optimal auction mechanism, whose goal is to maximize the SP's ER. Secondly, we keep the expected type of users $\frac{1}{2}(\o{a}+\u{a})$ fixed at 5 and let the length of the support $\o{a}-\u{a}$ vary from 0 to 9. In such a case, the level of uncertainty in user types increases while the expected types remain unaltered. The corresponding ER of the SP and the expected average utility of users are shown in Fig. \ref{ErU}-(b) and Fig. \ref{ErU}-(e), respectively. Unsurprisingly, the ER of the SP decreases with the length of the support of type distributions, since the uncertainty of user valuations hinders the SP's profit extraction. Moreover, as the length of support increases, the expected average utility of users first increases because the uncertainty in users' private valuations prevents the SP from charging too much payments. In particular, the expected average utility of users vanish when the length of the support is zero. In such a case, user types become deterministic so that the SP can reap all possible profits from the users. Interestingly, the expected average utility of users decreases slightly after the length of the support is large enough (c.f. Fig. \ref{ErU}-(e) when the length of support is larger than 7). The reason may be that, when user types are very uncertain, the SP cannot extract enough profit to compensate the content delivery and acquisition costs so that he often chooses not to cache anything. This in turn hurts the expected utility of users. Thirdly, we consider i.i.d. exponentially distributed user types with distribution parameters $\lambda_j=\lambda$ for all $j$. We let $\lambda$ increase so that the expected types ($1/\lambda$) decrease. The corresponding ER of the SP and the expected average utility of users are shown in Fig. \ref{ErU}-(c) and Fig. \ref{ErU}-(f), respectively. We observe that, as the expected user valuation decreases (i.e., $\lambda$ increases), both the ER of the SP and the expected average utility of users decrease. This is reasonable since high valuations of users benefit both the SP and the users themselves.

\begin{figure}
\renewcommand\figurename{\small Fig.}
\centering \vspace*{8pt} \setlength{\baselineskip}{10pt}
\subfigure[Impact of $\theta$ on the ER of the SP]{
\includegraphics[scale = 0.17]{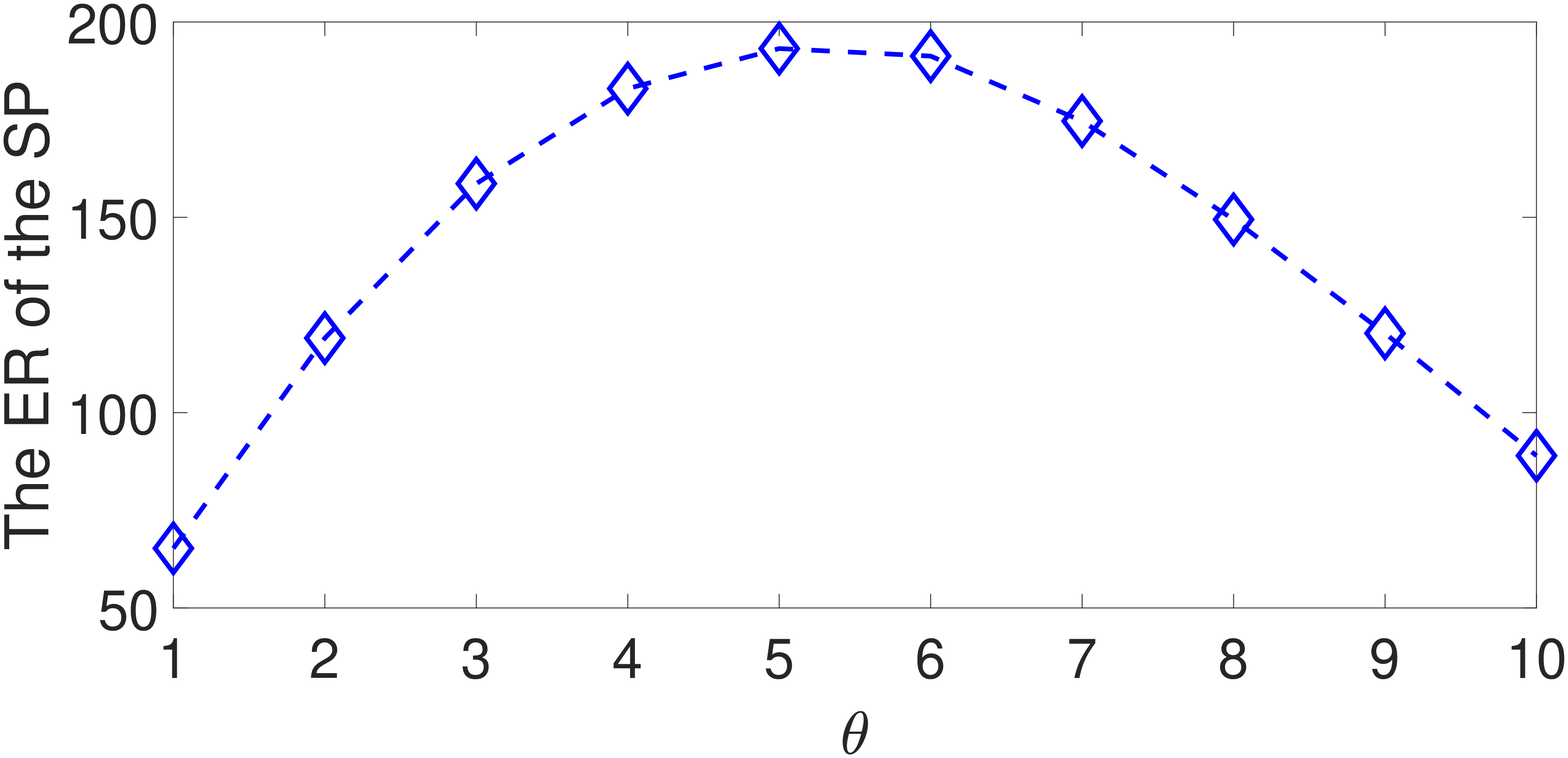}}
\subfigure[Impact of $\theta$ on the expected average utility of users]{
\includegraphics[scale = 0.17]{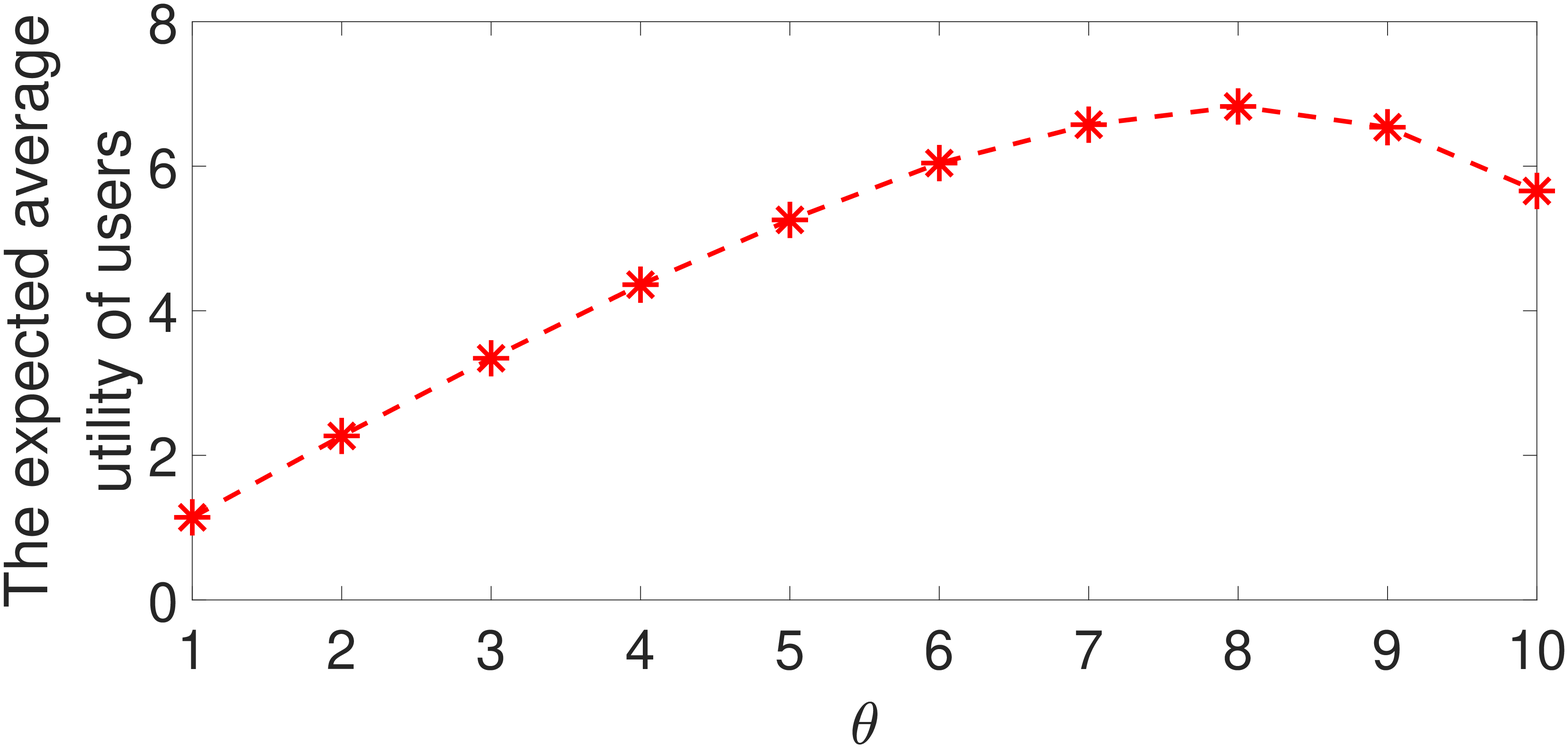}}
\caption{Impact of $\theta$ on the ER of the SP and the expected average utility of users}
\label{theta}
\end{figure}

Furthermore, we examine the impact of the content delivery quality $\theta$ on the SP and the users. Suppose user types are i.i.d. uniformly distributed on $[\u{a},\o{a}]$, where we set $\u{a}=1$ and $\o{a}=4$. Suppose we have a (relatively) large number of users $n=100$ and each $\b{\Omega}_i$ is formed according to the procedure in Assumption \ref{assump_formation} with $q_1=0.7$, $q_2=0.5$, $q_3=0.4$. It can be easily verified that the quadratic content delivery cost function $h(\theta)=\alpha\theta^2$ satisfies Assumption \ref{assump_cost_func}, wheree $\alpha$ is set to be $0.1$ here. Thus, Assumptions \ref{assump_large_user}-\ref{assump_cost_func} all hold true. According to Theorem \ref{thm_opt_theta}, the optimal $\theta^*$ that maximizes the ER of the SP is $\theta^*=\frac{\u{a}}{2\alpha}=5$.  To confirm this empirically, we report the ER of the SP and the expected average utility of users with varying $\theta$ in Fig. \ref{theta}-(a) and Fig. \ref{theta}-(b), respectively. From Fig. \ref{theta}-(a), we observe that, in accordance with our theoretical result, the optimal $\theta^*$ to maximize the ER of the SP is $5$. From \ref{theta}-(b), we further see that the value of $\theta$ that maximizes the expected average utility of users is $8$, which is larger than $5$, the optimal $\theta^*$ for the SP. This is reasonable since improving the delivery quality $\theta$ directly benefits users and only indirectly benefits the SP through users' payments. As such, users tend to prefer higher delivery quality than the SP does.

\begin{figure}
\renewcommand\figurename{\small Fig.}
\centering \vspace*{8pt} \setlength{\baselineskip}{10pt}
\subfigure[Impact of the content popularity on the ER of the SP]{
\includegraphics[scale = 0.17]{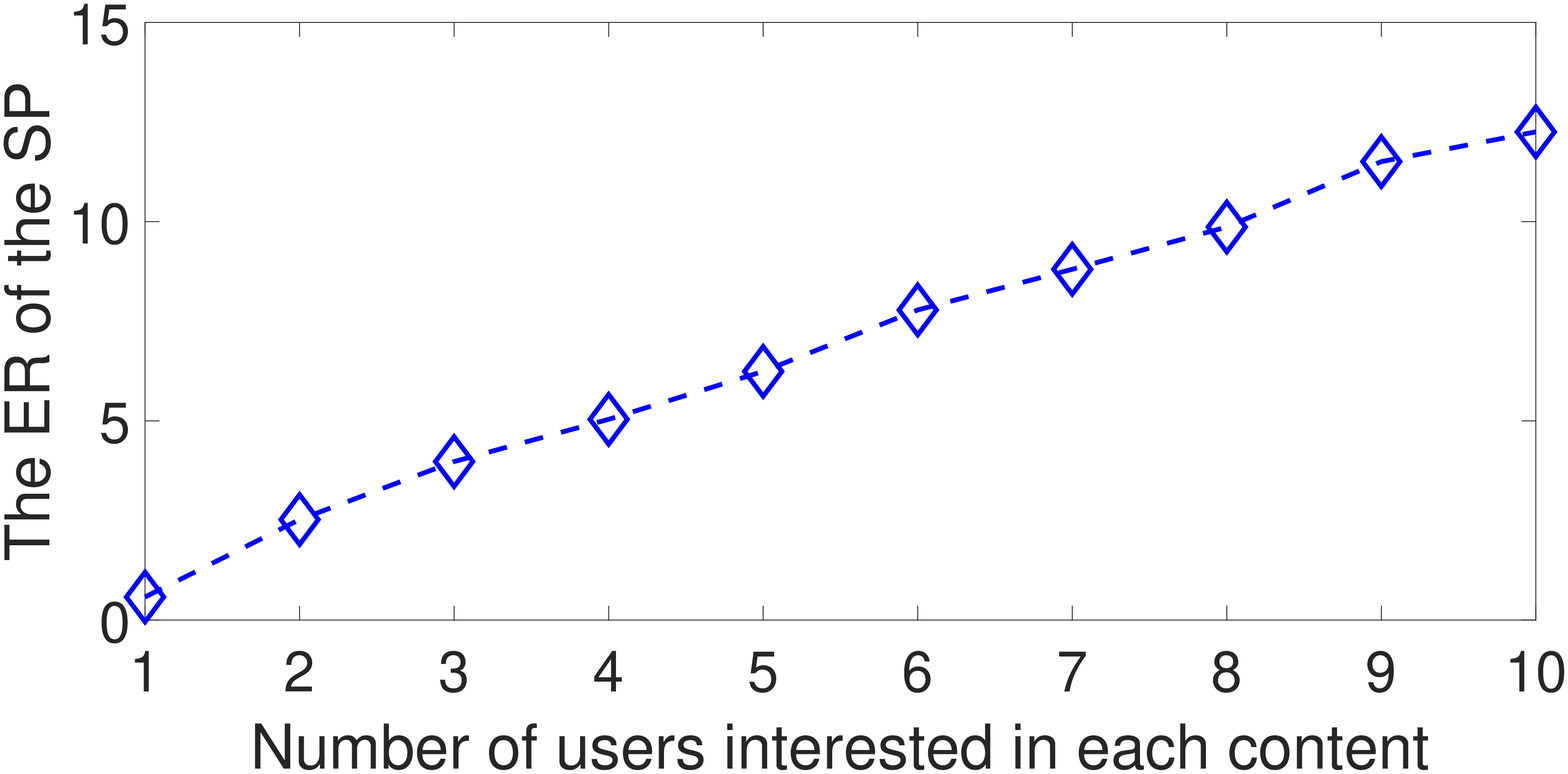}}
\subfigure[Impact of the content popularity on  the expected average utility of users]{
\includegraphics[scale = 0.17]{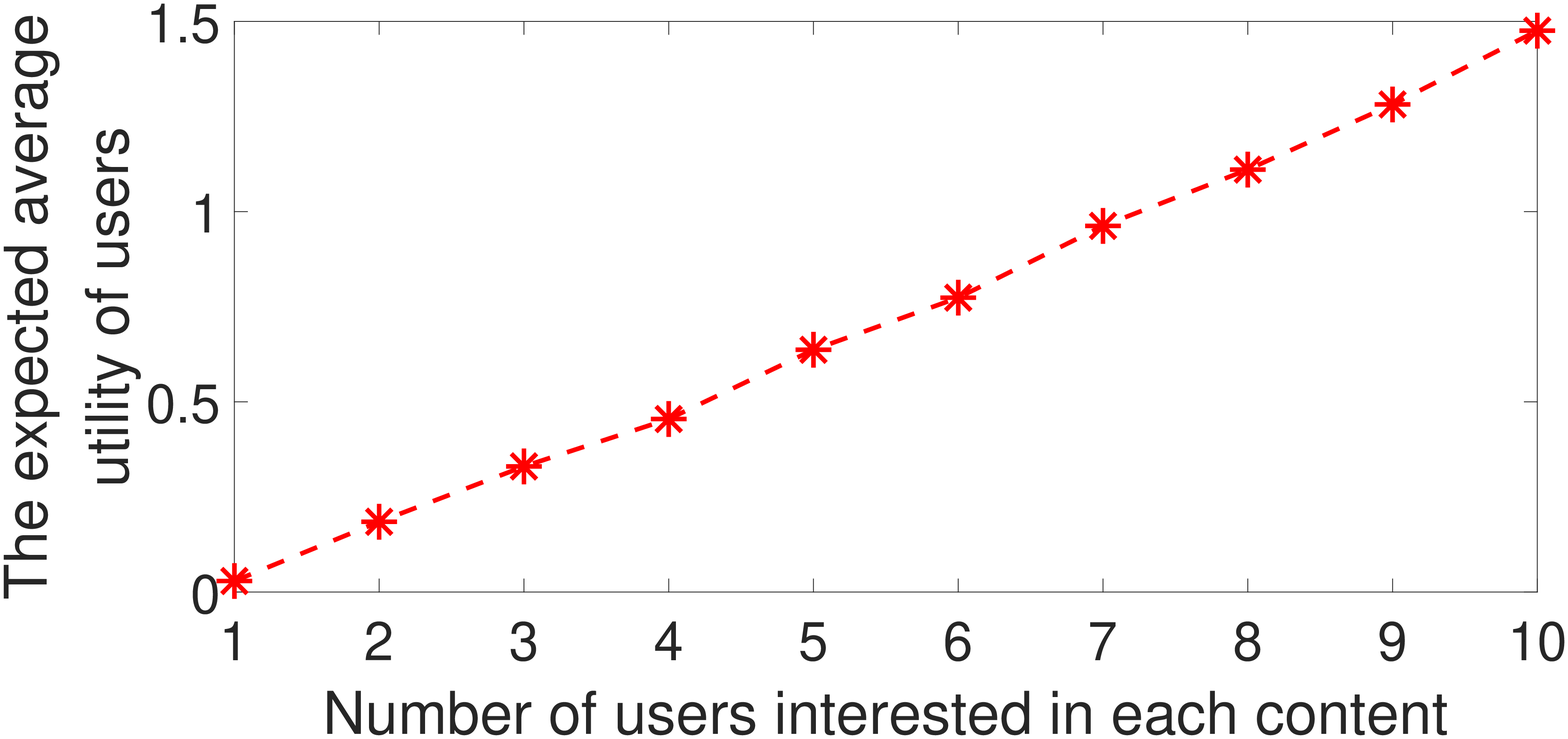}}
\caption{Impact of the content popularity (captured by $|\b{\Omega}_i|$)}
\label{omega}
\end{figure}

Next, we study the impact of content popularity, which is captured by $|\b{\Omega}_i|$. We return to the case of $n=10$ users whose types are uniformly distributed with parameters $\u{a}_j=1+0.1(j-1)$ and $\o{a}_j=4+0.1(j-1)$ for all $j$. For simplification, we let each content interest the same number of users $k$, i.e., $|\b{\Omega}|_i=k$ for all $i$, and each $\b{\Omega}_i$ consists of $k$ randomly picked users. The ER of the SP and the expected average utility of users with varying $k$, i.e., varying number of users interested in each content, are illustrated in Fig. \ref{omega}-(a) and Fig. \ref{omega}-(b), respectively. We observe that both the ER of the SP and the expected average utility of users increase with content popularity. As contents become more popular, the SP can reap profits from more users. On the other hand, each user is charged with lower price and her utility can be boosted.

\begin{figure}
\renewcommand\figurename{\small Fig.}
\centering \vspace*{8pt} \setlength{\baselineskip}{10pt}
\subfigure[Impact of the number of users on the ER of the SP]{
\includegraphics[scale = 0.17]{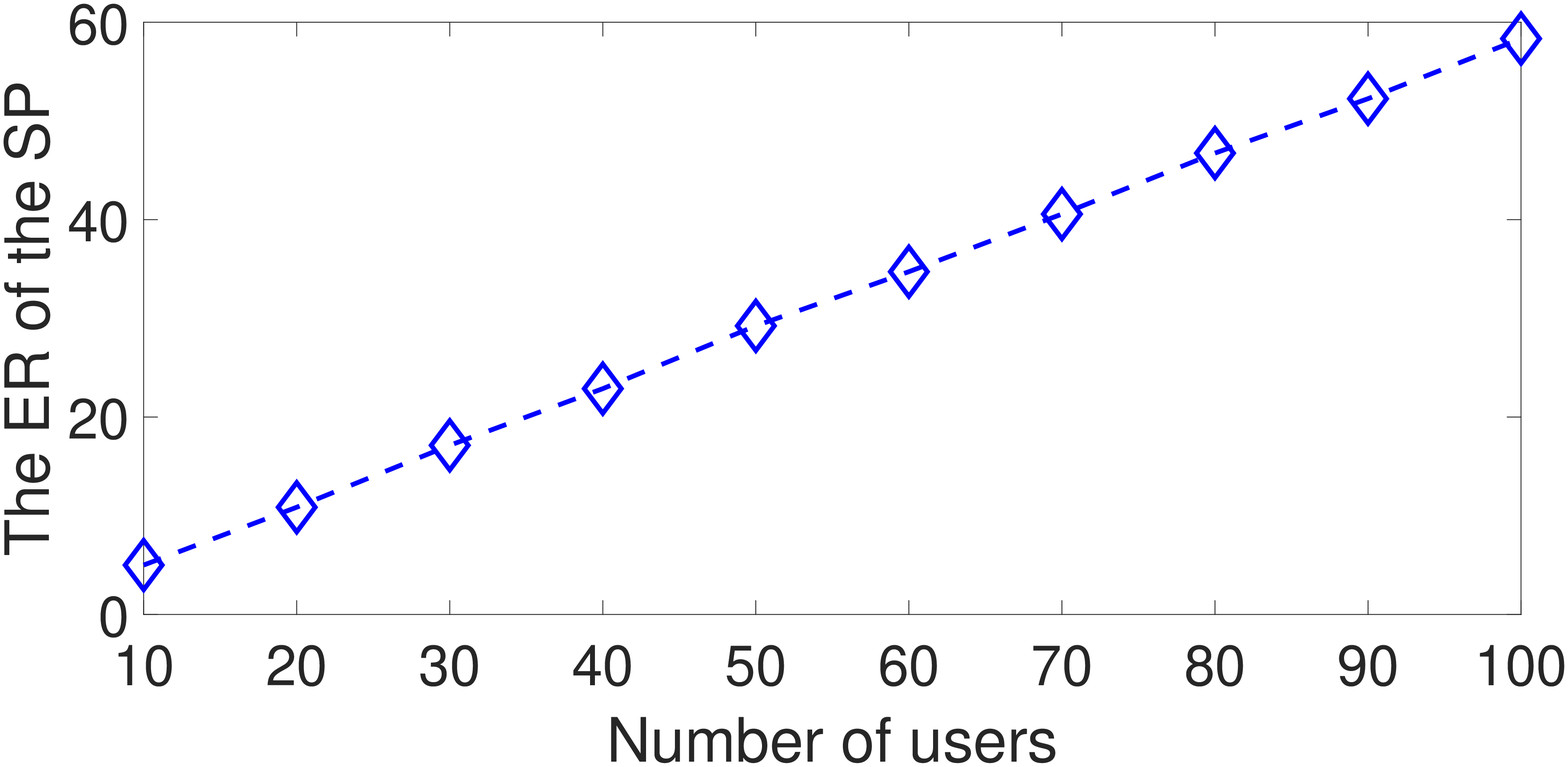}}
\subfigure[Impact of the number of users on the expected average utility of users]{
\includegraphics[scale = 0.17]{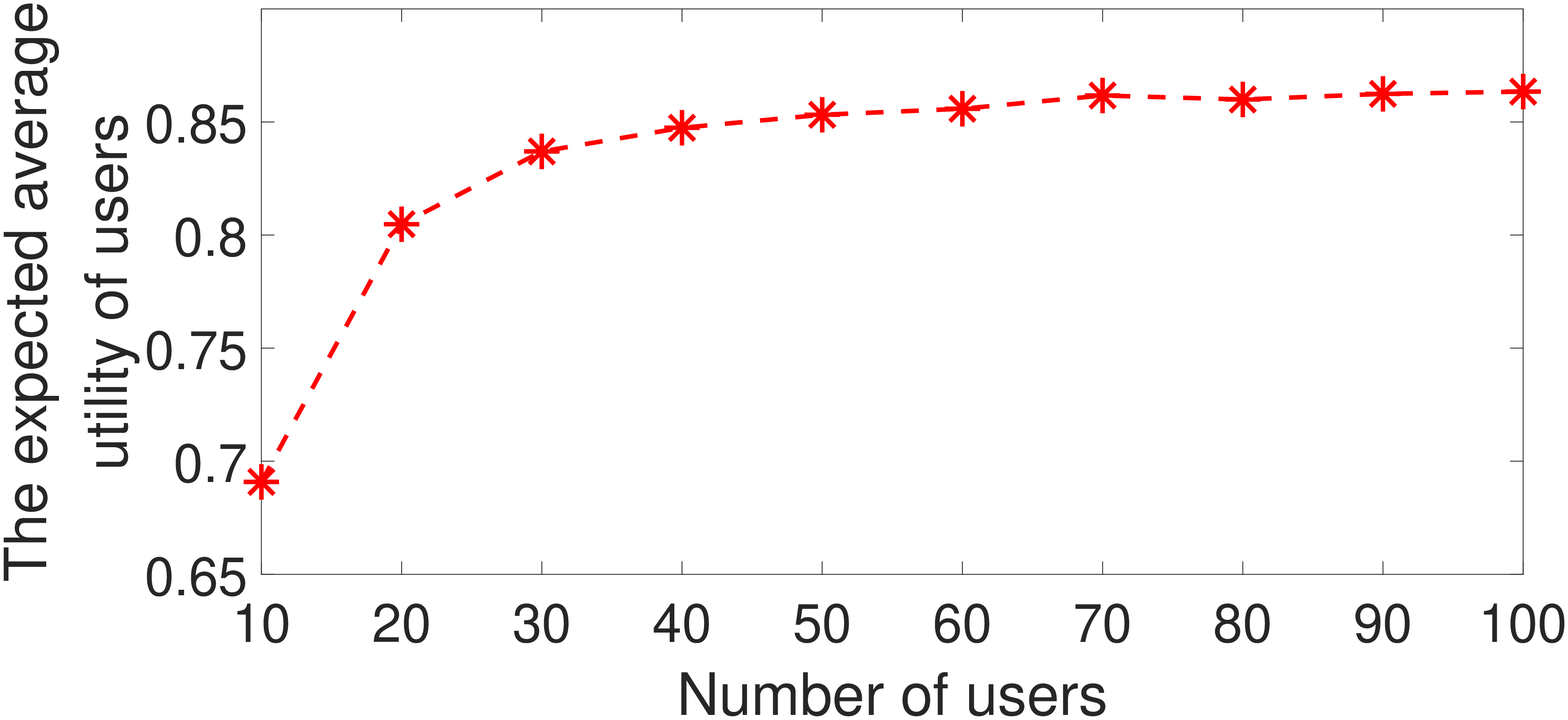}}
\caption{Impact of the number of users $n$}
\label{impact_numuser}
\end{figure}

We also examine the impact of the number of users $n$ on the ER of the SP and the expected average utility of users in Fig. \ref{impact_numuser}-(a) and Fig. \ref{impact_numuser}-(b), respectively. We set each $|\b{\Omega}_i|$ to be $0.6n$, i.e., each content interests a constant proportion of users. We observe that the ER of the SP increases with $n$ because the SP can collect more profits from more users. The expected average utility of users is increasing at first and gradually saturates for large $n$. This suggests that, in the optimal auction mechanism, the SP will not further discount its content prices if the number of users is already large enough.

\begin{figure}
\renewcommand\figurename{\small Fig.}
\centering \vspace*{8pt} \setlength{\baselineskip}{10pt}
\subfigure[Impact of content delivery cost function on the ER of the SP]{
\includegraphics[scale = 0.17]{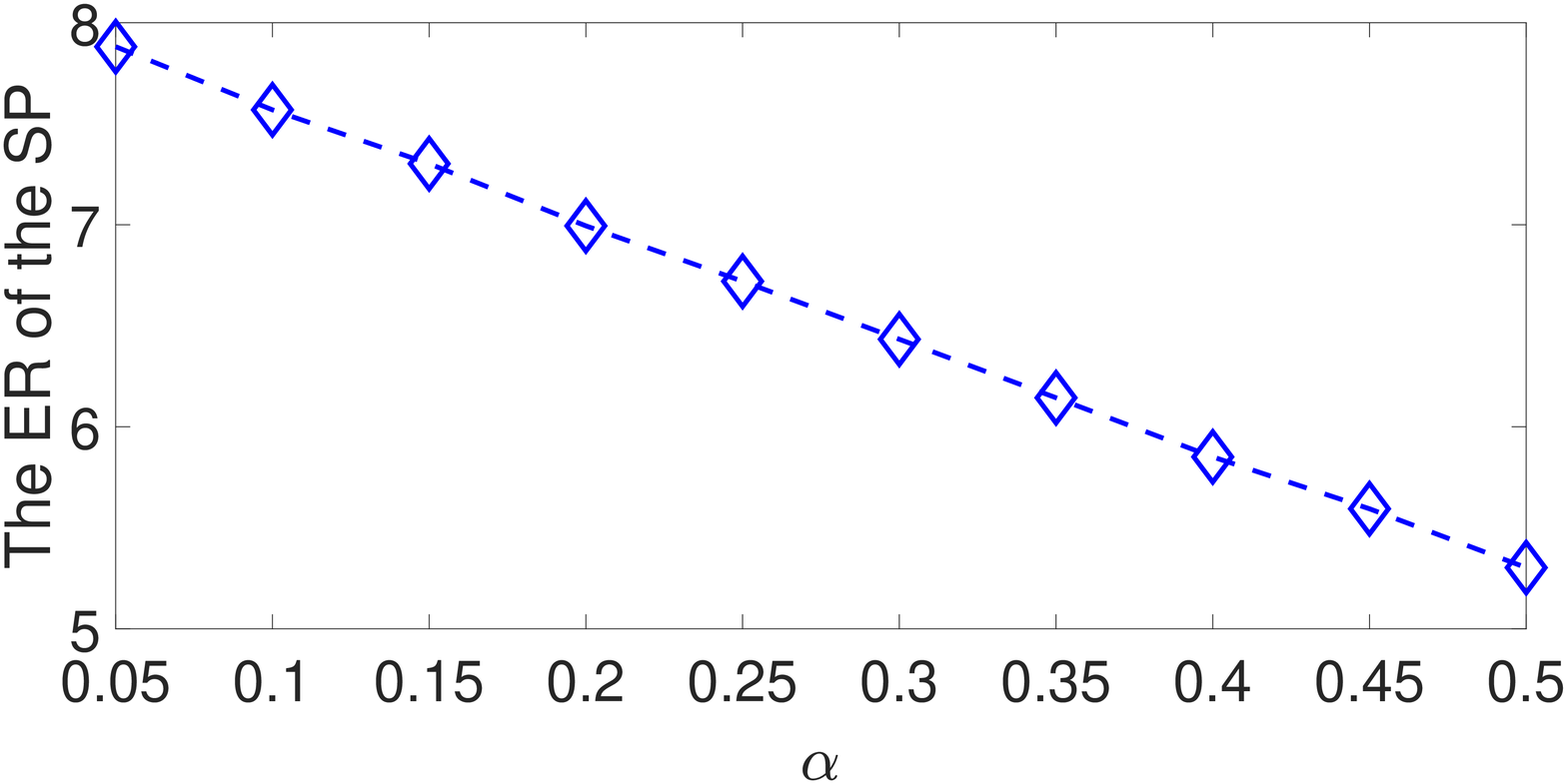}}
\subfigure[Impact of content delivery cost function on the expected average utility of users]{
\includegraphics[scale = 0.17]{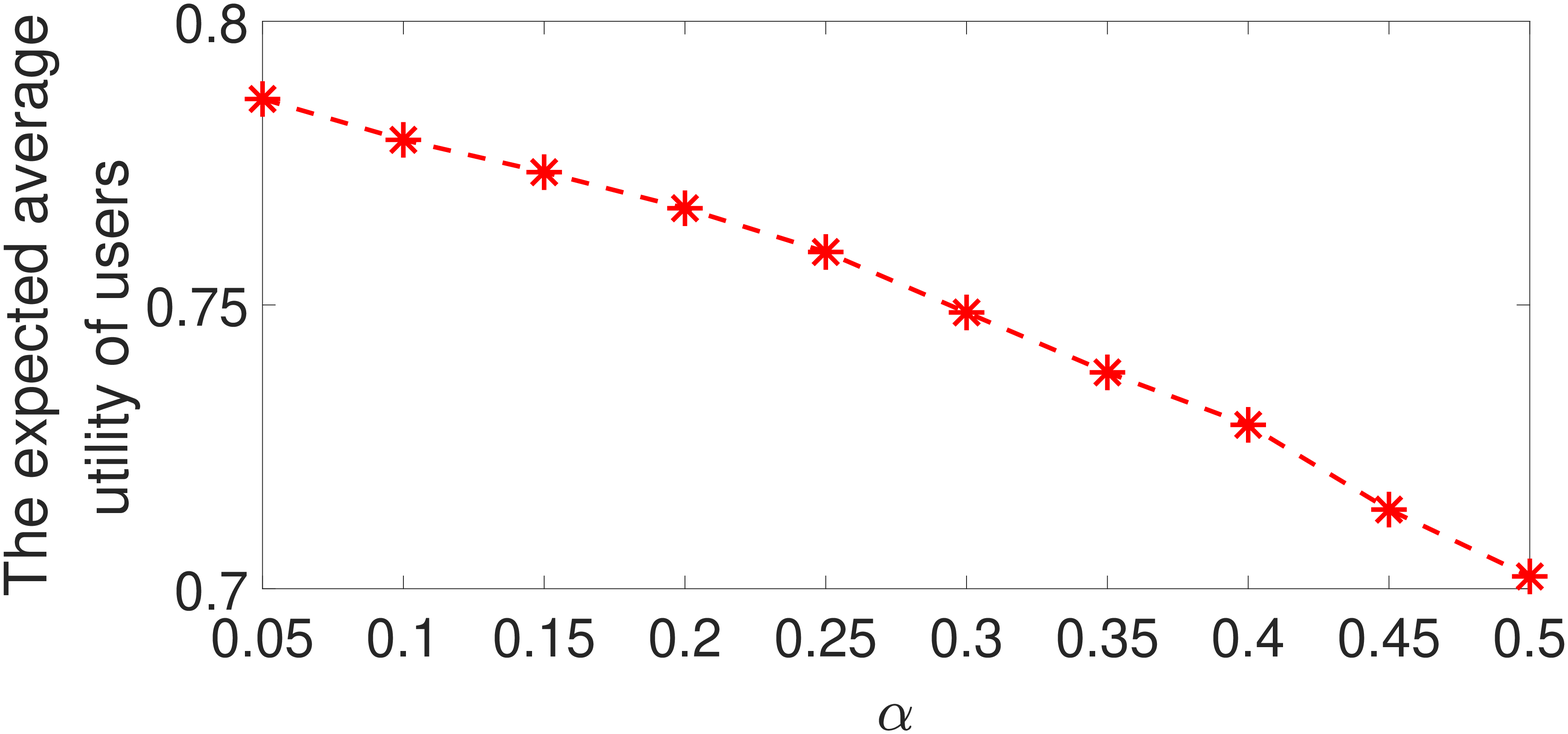}}
\caption{Impact of content delivery cost function (captured by $\alpha$)}
\label{alpha}
\end{figure}

In addition, we investigate the impact of content delivery cost function, which is controlled by the coefficient $\alpha$ in the quadratic cost function used here. The greater the $\alpha$ is, the higher the content delivery cost for the same delivery quality $\theta$, which is set to be $0.1$ here. We report the ER of the SP and the expected average utility of users in Fig. \ref{alpha}-(a) and Fig. \ref{alpha}-(b), respectively. We observe that, as $\alpha$ increases, both the ER of the SP and the expected average utility of users decrease while the former decreases faster. This is reasonable since increasing delivery cost directly hurts the SP and only reduces users' utility indirectly through increasing payments.

\begin{figure}
\renewcommand\figurename{\small Fig.}
\centering \vspace*{8pt} \setlength{\baselineskip}{10pt}
\subfigure[Uniformly distributed user types]{
\includegraphics[scale = 0.17]{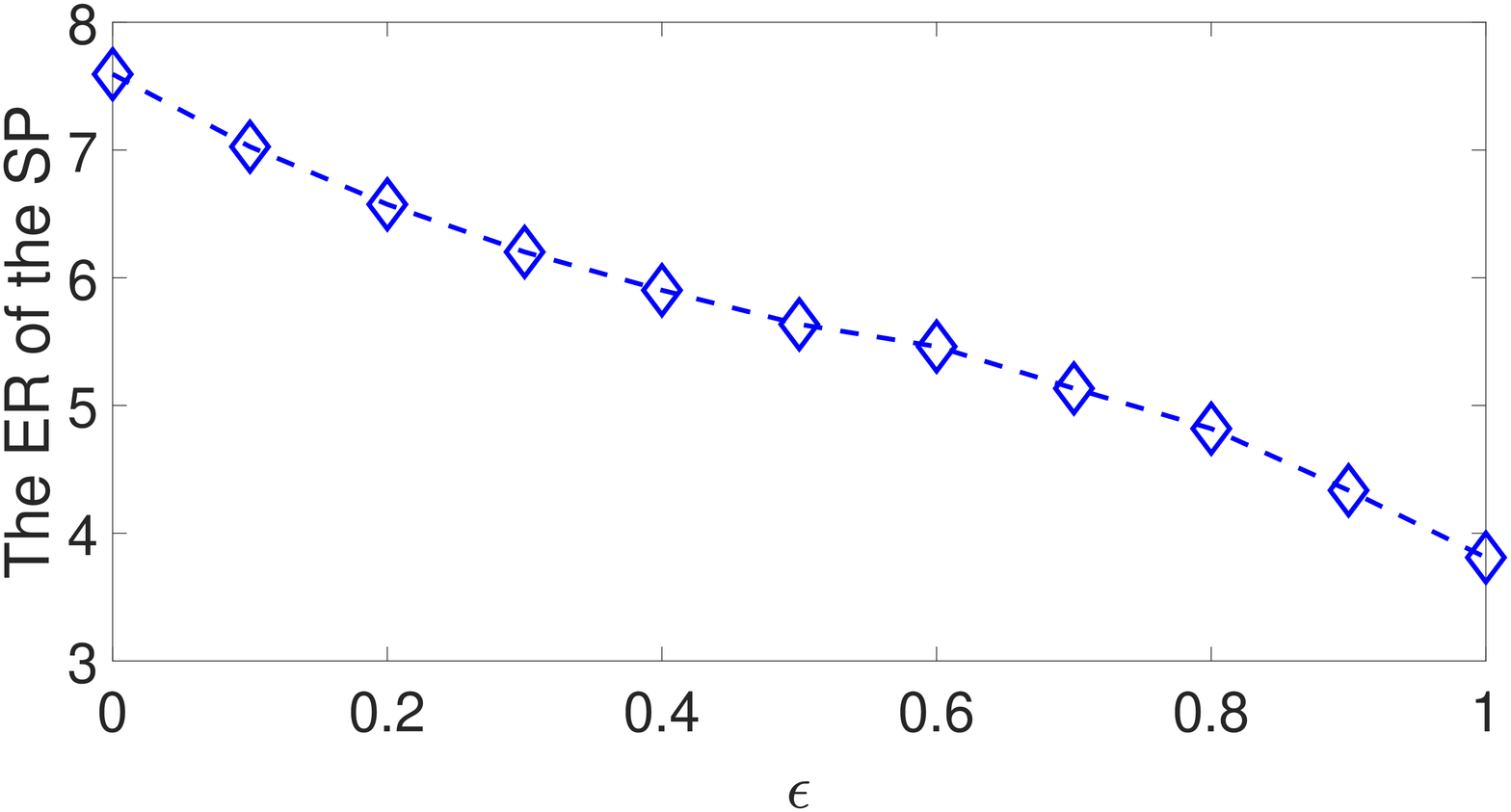}}
\subfigure[Exponentially distributed user types]{
\includegraphics[scale = 0.17]{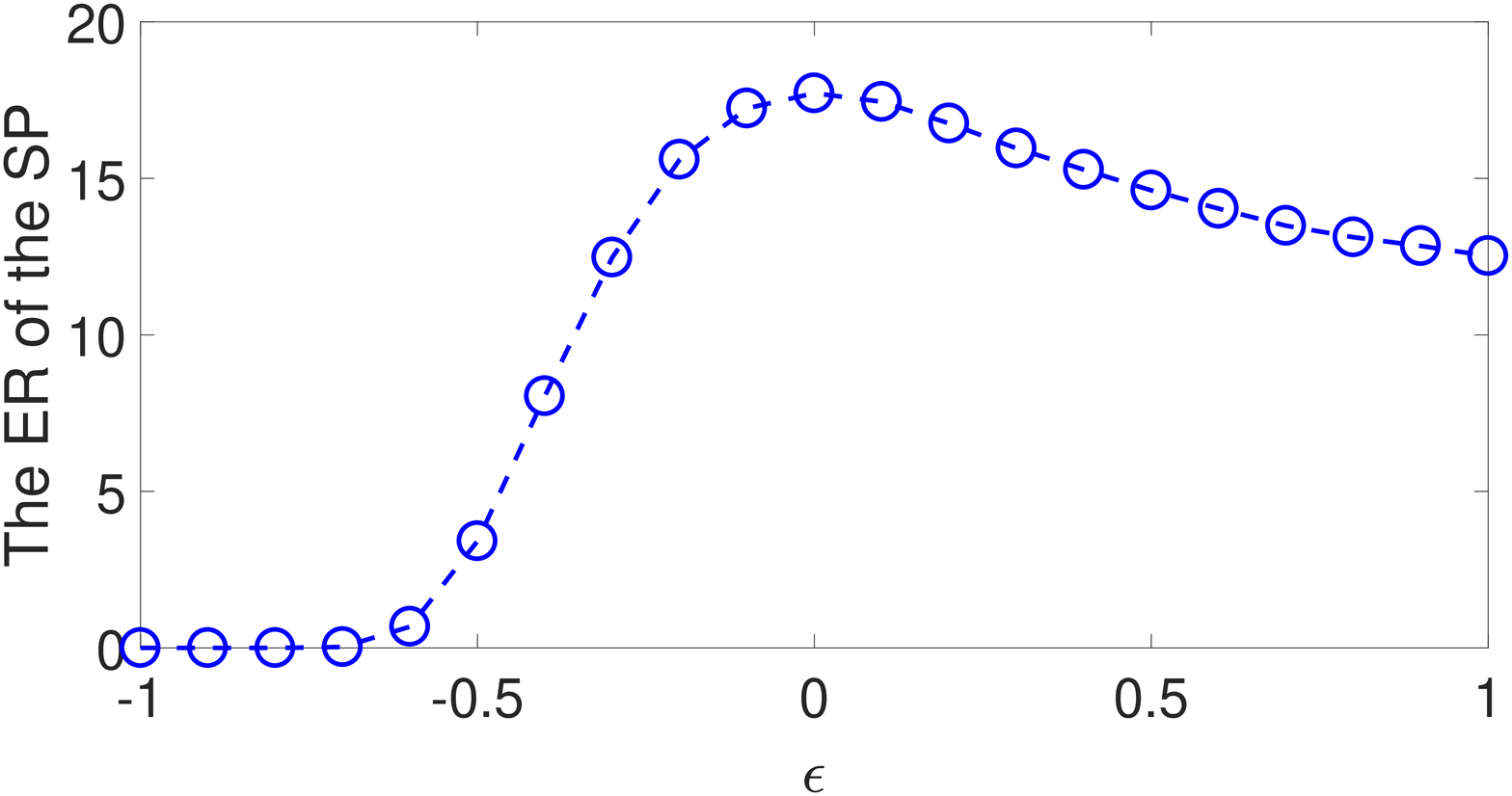}}
\caption{Impact of the modeling errors of user type distributions}
\label{error}
\end{figure}

Last but not least, we examin the impact of modeling errors of user type distributions. Suppose the true distribution of user $j$'s type is uniform distribution over the interval $[\u{a}_j^{\text{true}},\o{a}_j^{\text{true}}]$. Because of factors such as lack of past usage history, the SP may not be able to infer the true parameters $\u{a}_j^{\text{true}}$ and $\o{a}_j^{\text{true}}$ accurately. Instead, due to lack of information, the SP may choose to adopt more conservative estimates (thus a larger support of distributions) $\u{a}_j^{\text{est}}=\u{a}_j^{\text{true}}-\frac{\epsilon(\o{a}_j^{\text{true}}-\u{a}_j^{\text{true}})}{2}$ and $\o{a}_j^{\text{est}}=\o{a}_j^{\text{true}}+\frac{\epsilon(\o{a}_j^{\text{true}}-\u{a}_j^{\text{true}})}{2}$, where $\epsilon>0$ characterizes the relative estimation error. Then, the estimated $\u{a}_j^{\text{est}}$ and $\o{a}_j^{\text{est}}$ are fed to the proposed optimal auction mechanism to allocate the cache space and to decide the payments. In such a case, the ER of the SP with varying relative error $\epsilon$ is shown in Fig. \ref{error}-(a). We observe that the ER decreases smoothly with $\epsilon$ and the SP can keep at least $75\%$ of its ER as long as the relative error $\epsilon$ is no more than $0.5$. Moreover, we repeat the same experiment for exponentially distributed user types. The estimated distribution parameter is set to be $\lambda_j^{\text{est}}=(1+\epsilon)\lambda_j^{\text{true}}$, where $\lambda_j^{\text{true}}$ is the true distribution parameter for user $j$. In such a case, when $\epsilon$ varies from $-1$ to $1$, the corresponding ER of the SP is reported in Fig. \ref{error}-(b). Analogous to the scenario of uniformly distributed types, the ER still decreases with increasing relative error $|\epsilon|$ when types are exponentially distributed. We also observe that the ER decreases very slowly when $\epsilon$ is positive. Thus, if the SP cannot infer $\lambda_j$ precisely, he'd better overestimate it than underestimate it. In other words, if the SP is unclear about the true type distribution, he'd better underestimate the expected valuation ($1/\lambda_j$) of users than overestimate it.

\section{Conclusion}
In this paper, we have designed an optimal auction mechanism for content caching to maximize the expected revenue of the SP. The mechanism takes into consideration both the content acquisition costs and the content delivery costs. It incentivizes truthful reports of the private user types (incentive compatibility) and user participation of the mechanism (individual rationality). Moreover, computationally efficient methods of calculating the optimal cache space allocation and user payments have been presented and it has been shown that a user does not need to pay anything if no content of her interest is cached by the SP. We further examine the optimal choice of the content delivery quality under the reasonable hypothesis of large number of users and have derived a simple formula to compute the optimal delivery quality. Finally, extensive simulations have been carried out to evaluate the performance of the proposed mechanism and the impact of various model parameters has been highlighted.

\section*{Appendix A: Proof of Lemma \ref{lem_constraints}}
\textbf{[Necessity]} Suppose IC in \eqref{ic} and IR in \eqref{ir} hold. Then, for any $j=1,...,n$, $t_j,\tau_j\in\b{T}_j$:
\begin{align}
\w{v}_j(t_j)&\geq v_j(\tau_j,t_j)\\
&=\theta (t_j-\tau_j)\w{p}_j(\tau_j)+\theta\tau_j\w{p}_j(\tau_j)-\w{x}_j(\tau_j)\\
&=\theta (t_j-\tau_j)\w{p}_j(\tau_j)+\w{v}_j(\tau_j).\label{l1_1}
\end{align}
Symmetrically, we have:
\begin{align}
\w{v}_j(\tau_j)-\w{v}_j(t_j)\geq\theta(\tau_j-t_j)\w{p}_j(t_j).\label{l1_2}
\end{align}
Hence,
\begin{align}
\theta(\tau_j-t_j)\w{p}_j(\tau_j)\geq\w{v}_j(\tau_j)-\w{v}_j(t_j)\geq\theta(\tau_j-t_j)\w{p}_j(t_j).\label{l1_3}
\end{align}
Since $\theta>0$, from \eqref{l1_3}, we have:
\begin{align}
(\tau_j-t_j)(\w{p}_j(\tau_j)-\w{p}_j(t_j))\geq0,~\forall t_j,\tau_j\in\b{T}_j,
\end{align}
which implies $\w{p}_j(\cdot)$ is an increasing function, i.e., (i) holds. From \eqref{l1_3}, for $\tau_j>t_j$:
\begin{align}
\theta\w{p}_j(\tau_j)\geq\frac{\w{v}_j(\tau_j)-\w{v}_j(t_j)}{\tau_j-t_j}\geq\theta\w{p}_j(t_j).
\end{align}
Letting $\tau_j$ approaches $t_j$, we get:
\begin{align}
\lim_{\tau_j\downarrow t_j}\frac{\w{v}_j(\tau_j)-\w{v}_j(t_j)}{\tau_j-t_j}=\theta\w{p}_j(t_j).
\end{align}
It can be analogously shown that the left limit of the above relation also holds so that $\frac{d}{dt_j}\widetilde{v}_j(t_j)=\theta\w{p}_j(t_j)$, $\forall t_j\in\b{T}_j$. Integration of this relation leads to statement (ii). Statement (iii) clearly follows from IR.

\textbf{[Sufficiency]} In converse, suppose (i), (ii), (iii) hold. From the FF constraint, we know $\w{p}_j(\tau_j)$, $\forall j$ and $\forall \tau_j\in\b{T}_j$. Hence, from (ii) and (iii), we have $\w{v}_j(t_j)\geq\w{v}_j(\u{a}_j)\geq0$, i.e., IR holds. Furthermore, from \eqref{v2} and (ii), we can derive, for any $j=1,...,n$, $t_j\in\b{T}_j$:
\begin{align}
\w{x}_j(t_j)=\theta t_j\w{p}_j(t_j)-\w{v}_j(\u{a}_j)-\theta\int_{\u{a}_j}^{t_j}\w{p}_j(\tau_j)d\tau_j.\label{l1_4}
\end{align}
Therefore, for any $t_j,t_j'\in\b{T}_j$:
\begin{align}
&\w{v}_j(t_j)-v_j(t_j',t_j)\nonumber\\
&=\theta t_j\w{p}_j(t_j)-\w{x}_j(t_j)-\theta t_j\w{p}_j(t_j')+\w{x}_j(t_j')\\
&=\theta(t_j'-t_j)\w{p}_j(t_j')+\theta\int_{t_j'}^{t_j}\w{p}_j(\tau_j)d\tau_j\\
&=\theta\int_{t_j'}^{t_j}\left[\w{p}_j(\tau_j)-\w{p}_j(t_j')\right]d\tau_j\ageq0,
\end{align}
where (a) follows from the monotonicity of $\w{p}_j(\cdot)$ in (i). Thus, IC holds.

\section*{Appendix B: Proof of Lemma \ref{lem_er}}
For feasible mechanism $\langle\@{p}(\cdot),\@{x}(\cdot)\rangle$, according to Lemma \ref{lem_constraints}, we know that statements (i), (ii) and (iii) in Lemma \ref{lem_constraints} hold. From the definition of $\w{v}_j(t_j)$, we get, for any $t_j\in\b{T}_j$:
\begin{align}
\w{v}_j(t_j)=\int_{\b{T}_{-j}}\left[\left(\sum_{i\in\b{S}_j}p_i(\@{t})\right)\theta t_j-x_j(\@{t})\right]f_{-j}(\@{t}_{-j})d\@{t}_{-j}.
\end{align}
Along with statement (ii) of Lemma \ref{lem_constraints}, we obtain:
\begin{align}
&\int_{\b{T}_{-j}}x_j(\@{t})f_{-j}(\@{t}_{-j})d\@{t}_{-j}\nonumber\\
&=\int_{\b{T}_{-j}}\left[\sum_{i\in\b{S}_j}p_i(\@{t})\right]\theta t_jf_{-j}(\@{t}_{-j})d\@{t}_{-j}-\w{v}_j(\u{a}_j)\nonumber\\
&~~~-\theta\int_{\u{a}_j}^{t_j}\w{p}_j(\tau_j)d\tau_j\\
&\aeq-\w{v}_j(\u{a}_j)+\theta\int_{\b{T}_{-j}}\Bigg\{\left[\sum_{i\in\b{S}_j}p_i(\@{t})\right]t_j\nonumber\\
&~~~-\int_{\u{a}_j}^{t_j}\sum_{i\in\b{S}_j}p_i(\tau_j,\@{t}_{-j})d\tau_j\Bigg\}f_{-j}(\@{t}_{-j})d\@{t}_{-j},
\end{align}
where in (a) we make use of the definition of $\w{p}_j(\tau_j)$ in \eqref{p_tilde_def} and interchange the order of integrals. Performing the operation $\int_{\b{T}_j}(\cdots)f_j(t_j)dt_j$ on both sides of the above equation, we obtain:
\begin{align}
&\int_{\b{T}}x_j(\@{t})f(\@{t})d\@{t}\nonumber\\
&=-\w{v}_j(\u{a}_j)+\theta\int_\b{T}\left[\sum_{i\in\b{S}_j}p_i(\@{t})\right]t_jf(\@{t})d\@{t}\nonumber\\
&~~~-\theta\int_{\b{T}_{-j}}\left\{\int_{\b{T}_j}\int_{\u{a}_j}^{t_j}\left[\sum_{i\in\b{S}_j}p_i(\tau_j,\@{t}_{-j})\right]f_j(t_j)d\tau_jdt_j\right\}\nonumber\\
&~~~\cdot f_{-j}(\@{t}_{-j})d\@{t}_{-j}\\
&\aeq-\w{v}_j(\u{a}_j)+\theta\int_\b{T}\left[\sum_{i\in\b{S}_j}p_i(\@{t})\right]t_jf(\@{t})d\@{t}\nonumber\\
&~~~-\theta\int_{\b{T}_{-j}}\int_{\b{T}_j}[1-F_j(t_j)]\left[\sum_{i\in\b{S}_j}p_i(\@{t})\right]dt_jf_{-j}(\@{t}_{-j})d\@{t}_{-j}\\
&=-\w{v}_j(\u{a}_j)+\theta\int_\b{T}\left[t_j-\frac{1-F_j(t_j)}{f_j(t_j)}\right]\left[\sum_{i\in\b{S}_j}p_i(\@{t})\right]f(\@{t})d\@{t},\label{l2_1}
\end{align}
where in (a) we interchange order of integrals:
$\int_{\b{T}_j}\int_{\u{a}_j}^{t_j}\left[\sum_{i\in\b{S}_j}p_i(\tau_j,\@{t}_{-j})\right]f_j(t_j)d\tau_jdt_j
=\int_{\b{T}_j}\int_{\tau_j}^{\o{a}_j}\left[\sum_{i\in\b{S}_j}p_i(\tau_j,\@{t}_{-j})\right]f_j(t_j)dt_jd\tau_j
=\int_{\b{T}_j}[1-F_j(t_j)]\left[\sum_{i\in\b{S}_j}p_i(\@{t})\right]dt_j$.

Substituting \eqref{l2_1} into \eqref{er}, we get:
\begin{align}
\texttt{ER}&=-\sum_{j=1}^n\w{v}_j(\u{a}_j)\nonumber\\
&~~~+\theta\int_\b{T}\left\{\sum_{j=1}^n\left[t_j-\frac{1-F_j(t_j)}{f_j(t_j)}\right]\left[\sum_{i\in\b{S}_j}p_i(\@{t})\right]\right\}f(\@{t})d\@{t}\nonumber\\
&~~~-\int_\b{T}\left\{\sum_{i=1}^mp_i(\@{t})\left[r_i+|\b{\Omega}_i|h(\theta)\right]\right\}f(\@{t})d\@{t}\\
&\aeq-\sum_{j=1}^n\w{v}_j(\u{a}_j)+\theta\int_\b{T}\Bigg\{\sum_{j=1}^n\left[t_j-\frac{1-F_j(t_j)}{f_j(t_j)}-\frac{h(\theta)}{\theta}\right]\nonumber\\
&~~~\cdot\left[\sum_{i\in\b{S}_j}p_i(\@{t})\right]\Bigg\}f(\@{t})d\@{t}-\int_\b{T}\left[\sum_{i=1}^mp_i(\@{t})r_i\right]f(\@{t})d\@{t},
\end{align}
where in (a) we make use of $\sum_{j=1}^n\sum_{i\in\b{S}_j}p_i(\@{t})=\sum_{i=1}^mp_i(\@{t})|\b{\Omega}_i|$.

\bibliography{mybib}{}
\bibliographystyle{ieeetr}

\end{document}